\documentclass{article}

\usepackage[margin=0.9in, a4paper]{geometry}
\usepackage[english]{babel}
\usepackage[T1]{fontenc}

\usepackage{
	amsmath,
	amssymb,
	amsthm,
	bbm,
	enumitem,
	float,
	graphicx,
	hyperref,
	cleveref,
	inputenc,
	mathrsfs,
	parskip,
	subcaption,
	subfiles
}

\renewcommand{\P}{\mathbb{P}}
\newcommand{\N}{\mathbb{N}}
\newcommand{\R}{\mathbb{R}}
\newcommand{\E}{\mathbb{E}}
\newcommand{\I}{\mathbb{I}}
\newcommand{\one}{\mathbbm{1}}
\newcommand{\norm}[1]{\left\lVert#1\right\rVert}

\theoremstyle{plain}
\newtheorem{theorem}{Theorem}[section]
\newtheorem{proposition}[theorem]{Proposition}

\theoremstyle{definition}
\newtheorem{definition}[theorem]{Definition}

\begin{document}

\title{Interpolating Location Data with Brownian Motion}
\author{Ludo Dekker, Kerim \DJ eli\'c, Maarten van Dijk, Sven Holtrop, \\ Noah Keuper, Lizanne van der Laan, Tess van Leeuwen, \\ Caspar Meijs, Hanneke Schroten, Lisanne van Wijk}
\date{April 20, 2022}
\maketitle

\noindent\rule{\linewidth}{0.5mm}

\begin{abstract}
	In 2018 the ``Onderweg in Nederland'' (translation: the Dutch Travel Survey) project by Statistics Netherlands was commenced, where participants were asked to track their location using an app, and check if the data is correct. An issue that occurs is so-called `gaps' in the location data, where a whole sequence of data points is missing. The easiest way to fill such a gap is with a straight line, but this leads to systematic errors such as an underestimation of the distance travelled. A more realistic way to fill this gap is with a stochastic process. We use a Brownian bridge to model the movement of a traveller, as these have been used before successfully in ecological research. We find an explicit expression for the distance travelled in terms of some parameters that can be obtained from the data. To test whether this method gives an accurate estimation of the distance travelled, we simulate travel data using multiple different processes. 
\end{abstract}

\clearpage
\tableofcontents

\clearpage
\section{Introduction}\label{sec:introduction}

In 2019 the Dutch population travelled a total distance of well over 200 billion kilometres within the Netherlands \cite{vervoersprestatie}, or equivalently, 250.000 round trips to the moon \cite{afstandmaan}. To be able to construct and maintain the right infrastructure for all these travel movements, it is very important to have reliable data regarding the travel habits of the Dutch population. Statistics Netherlands (CBS from the Dutch name: Centraal Bureau voor de Statistiek) aims to use travel surveys based on Mobile Device Location Tracking (MDLT), which includes GPS and other location tracking methods, to study the mobility of the Dutch population. Respondents use an app that tracks their location over a period of seven days. At the end of each day, they are asked to check whether the app has detected stops correctly and predicted the right mode of travel (for example: walking, cycling or taking the train). Using MDLT has many benefits over the more traditional way of mobility studies. Traditionally respondents were asked to write down all their travel movements at the end of the day. This leads to an underestimation of travel distance, number of trips and travel time, as respondents do not recollect every movement \cite{mccool2021app}. Writing down all travel movements is also more time consuming for respondents than verifying the data already collected by the app.

The usage of MDLT in travel surveys is a relatively new development and doing so at a national scale is unique \cite{smeets2019automatic}. However, ecologists have been using GPS trackers for a long time to study the territories or migrational routes of animals \cite{animalBBMM}.  Compared to the MDLT travel surveys of humans, the GPS trackers used on animals have a low sample rate: the CBS app measures the location of a moving participant every second, while animal GPS trackers typically measure the location every few minutes or even every few hours.

A difficulty with the use of MDLT travel surveys is that there can be gaps of missing data points. These gaps are periods of time during which no location data is collected.  The reasons for these gaps vary, ranging from loss of connection due to a tunnel or bridge, to the respondents phone shutting down the app, or the phone running out of battery. As there are multiple reasons for missing data, the gaps can vary from relatively small, perhaps only a few data points, to gaps of thousands of consecutive data points. Filling in the missing data perfectly is not possible, though by using an approximation it is still possible to get a good idea of the travel habits of the population. However, the missing data cannot be ignored. 

Hence, it is necessary to find a way to properly describe travel habits, in particular the distance travelled and in what area an individual is travelling. To quantify these travel habits we look at the length of a path taken by the traveller to quantify distance travelled and the so called Radius of Gyration (RoG) to quantify the area an individual is travelling through. The latter is a time weighted measure of the spread of the location of a person during a certain period of time. A low RoG indicates that a person spends most of their time in more or less the same area or that their travel paths are "bunched up", while a high RoG indicates that a person travels great distances. To fill the missing data we are now interested in making a good estimate for the missing data with respect to the above attributes of the data. A very simple model would be filling the gap with a straight line path, but it has some disadvantages. For example, it can significantly underestimate the distance travelled, as people probably did not travel in an exact straight line. A gap of many missing data points may also contain one or several stops. Conversely in such a gap a small round trip, for example to walk your dog, can be missed completely.

Our aim is to contribute to methods that can better estimate relevant attributes of the missing data. To do this we will be using a stochastic model, called a Brownian bridge, to approximate the missing paths.

We will start this report by introducing our ideas, experiments, and models in Chapter \ref{sec:methods}. Then we will give the background theory to understand the specific models we are using in Chapter \ref{sec:theory}. A detailed explanation of the ideas, experiments, and models that were briefly discussed before, can be found in Chapter \ref{sec:model}. We then present the results of our experiments in Chapter \ref{sec:simulation}, which we discuss afterwards in Chapter \ref{sec:discussion}. Then in the end we formulate a conclusion based on those results in Chapter \ref{sec:conclusion}.

\clearpage
\section{Methods}\label{sec:methods}

To fill up gaps in the data, we generally want to assume that a person moves in accordance with some stochastic process. Suppose that we have a missing segment in some movement data, then we want to use the available data to predict the movement in the gap. We will assume that the movement in the gap is realised by the same stochastic process with the same parameters as in the available data. Therefore, we now want to characterise our process in case we know it starts at a specific point and ends at a specific point. A stochastic process conditioned on its endpoint is called a bridge.

\subsection{Approaches}

The first of the processes we will discuss is Brownian motion, the continuous limit of a random walk. We looked into this for two reasons. The first and main reason is that Brownian motion is easy to do calculations with. This means it is possible to get theoretical results and to obtain good estimators for model parameters.
The second reason is that on average Brownian motion usually captures the dynamics of a process very well. When given two coordinates and a time between them, one can then construct the so-called Brownian bridge. 

This concept has already been used in ecology to model the movement of animals \cite{animalBBMM}. Now, whereas ecologists are interested in finding the habitat of the animal based on limited data points, we primarily want to know the path length, i.e. the distance travelled. However, as Brownian motion is nowhere differentiable, its path length is not well-defined \cite{kallenberg1997foundations}. To solve this problem, we discretise the Brownian bridge, so we get a random walk of which we can calculate the path length. We have found an expression for the expected value \eqref{eq:pathlengthexpectation} of the distance travelled for such a discretised Brownian bridge.

The second type of process we have briefly looked into, are Feller processes. This is a special type of continuous time process. For the interested reader, see \cite{kallenberg1997foundations}. 
Specifically, we considered a type of the Feller process, which consists of a position and an internal state.
The movement of the process is determined by the internal state, which is independent of the actual position of the process.
The internal state can be used to model the velocity of the person or different modes of transportation. 
In such a case, transition probabilities of the internal state represent changing one's modes of transportation, direction or speed.
For example, we say that someone who is stationary has a certain probability to remain stationary or to start moving again.
Similarly, for motion, we expect that the velocity of a person is dependent on their previous velocity, i.e. we do not expect someone to suddenly turn 180 degrees at high speed.
It turns out that calculating a bridge for such processes was too difficult for us. 
So, for filling gaps in the location data, we have only considered the Brownian bridge model. 
However, we have still simulated processes with an internal state numerically to assess how well the Brownian bridge model performs on motion that does not behave precisely as Brownian motion. All the stochastic models that we are going to describe, besides the Brownian bridge model, are purely for data generating purposes which will be more clear later on.

\subsection{Implementation}\label{subsec:implementation}

Note that we have briefly mentioned two type of models. For a formal definition of these models and their properties, the reader can look at Chapters \ref{sec:theory} and \ref{sec:model}. The aim of this section is to give a brief overview of our implementation of these models and how we have set up our experiments. Each model and experiment is implemented in Python. For the reader who is interested in our code, we refer to our \href{https://github.com/kerimdelic/MathForIndustry/blob/main/CBS.ipynb}{notebook}. The results of our experiments will be displayed in Chapter \ref{sec:simulation}.

First, we have implemented the Brownian bridge model. The parameters of our model are the following:
\begin{itemize}
    \item The starting position of the bridge. 
    \item The ending position of the bridge. 
    \item The number of time steps. 
    \item The \emph{diffusion coefficient} $\sigma_m$. 
\end{itemize}
With the above parameters, we can simulate a Brownian bridge and choose where it will be located. Furthermore, we can also choose what its deviation will be from the mean by fixing the diffusion coefficient $\sigma_m$. This coefficient will thoroughly be discussed in Section \ref{subsec:brownianbridge}. Another important parameter is the number of time steps, as this number is the same as the total number of location points that form the path. 

The model itself basically simulates a path that follows a Brownian motion. Of course, we can calculate certain properties of this path. The two properties that we will see in several parts throughout the paper are the path length and the RoG. The formal definition of the RoG can be found in Section \ref{subsec:rogtheory}. These properties can be extracted from a complete path, meaning that for a given path without gaps, we can numerically calculate the path length and the RoG. However, as we will discuss paths with a gap, we want to extract the same properties by filling the gap by a certain procedure. 

We followed the procedure as outlined in \cite{animalBBMM} to be able to estimate the diffusion coefficient. More precisely, it is of interest to estimate the mobility of a person during the travel. This description coincides with the definition of the diffusion coefficient. Thus, this procedure actually derives an estimate of this coefficient.

\subsubsection*{The Estimation Procedure}

Suppose the location data is of the form $(z_0, t_0), \dots, (z_{2n}, t_{2n})$. At each time $t_i$, there is a location $z_i$. The estimation procedure can be described by three steps.
\begin{enumerate}
    \item Between every other point we construct independent Brownian bridges with the same unknown parameter $\sigma_m$, so there is a bridge between $z_0$ and $z_2$, then a bridge between $z_2$ and $z_4$, and so on. 
    \item We consider each skipped point $z_{2k+1}$ as a realisation at time $t_{2k+1}$ of the Brownian bridge between $z_{2k}$ and $z_{2k+2}$.
    \item Then, we can use maximum likelihood estimation for these odd observations $(z_1, z_3, z_5,...)$ to estimate the unknown parameter $\sigma_m$. The explicit likelihood function can be found in \ref{eq:likelihood}. 
\end{enumerate}
The difference between our procedure and the procedure in the paper \cite{animalBBMM} is that we maximise the log-likelihood function. Besides this, we also use a different algorithm for finding the maximum of the function, namely ternary search.

An illustration of the estimation procedure is given in Figure \ref{fig:graph1}. 
\begin{figure}[h!]
    \centering
    \includegraphics[width=0.4\textwidth]{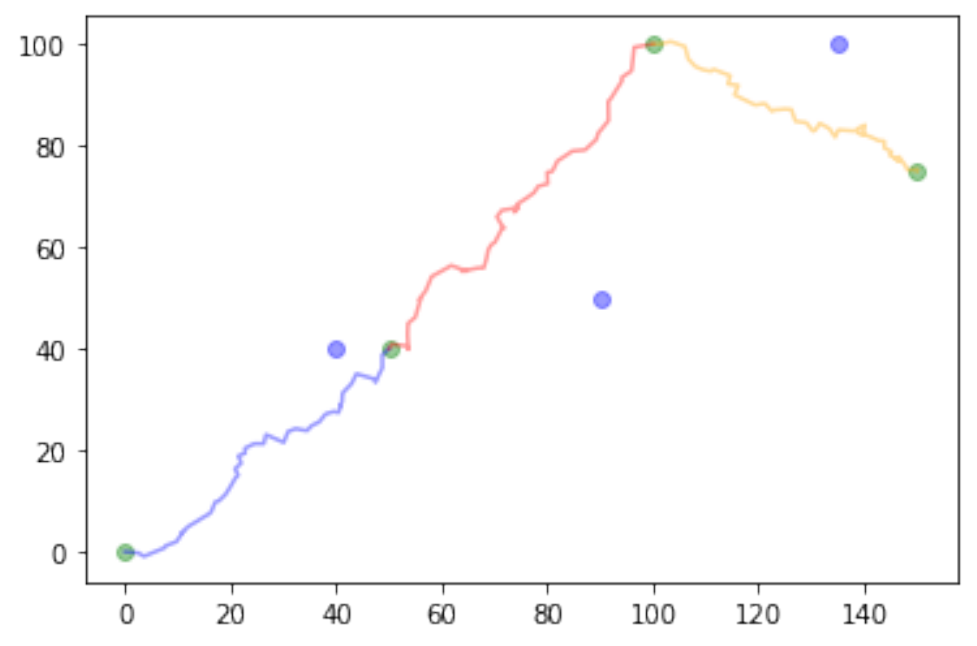}
    \caption{We have created seven ``fake'' points. The green dots are the even location points $(z_0, z_2, z_4, z_6)$ and the blue dots are the odd location points $(z_1, z_3, z_5)$. Between each pair of green dots, a Brownian bridge is simulated. All three bridges are modelled independently and form a large path. Each blue dot is interpreted as a location that results from the corresponding Brownian bridge. This location does not have to be on the path, and is used for our estimation procedure that obtains the most likely $\sigma_m$ for the three Brownian bridges.}
    \label{fig:graph1}
\end{figure}

Given some movement data with a gap, we can now estimate the diffusion coefficient. Moreover, we can also extract the starting and ending position of the gap and the number of time steps of this gap. This means we have all the parameters for the Brownian bridge model, and therefore we can calculate the expected path length during the gap. This expected value as a function of the parameters will be given in Section \ref{subsec:pathlength}. 

For the RoG we were not able to find an analytical expression for the expected value. However, since we have all the parameters for the Brownian bridge model we can simulate paths for the gap, and then calculate the RoG. To estimate the RoG we simulate this numerous times and take the average RoG over all simulations.

\subsection{Experiments}\label{subsec:experiments}

We have introduced our estimation procedure for the diffusion coefficient, whose accuracy we want to quantify in terms of path length and RoG.
Therefore, we have constructed two main experiments. One experiment concentrates on the path length and the other experiment is based on the RoG.
The general setup of these experiments is as follows:
\begin{enumerate}
    \item A path is created from some process.
    \item The actual values of the path length and RoG are calculated. 
    \item A gap is created, by deleting a certain number of consecutive location points.
    \item We estimate $\sigma_m$ by feeding our estimation procedure the incomplete path.
    \item With the estimated $\sigma_m$, we also estimate the path length and RoG.
\end{enumerate}
In Figure \ref{fig:procedure} an illustration of the above framework is given.

\begin{figure}[h!]
     \centering
     \begin{subfigure}[b]{0.3\textwidth}
         \centering
         \includegraphics[width=\textwidth]{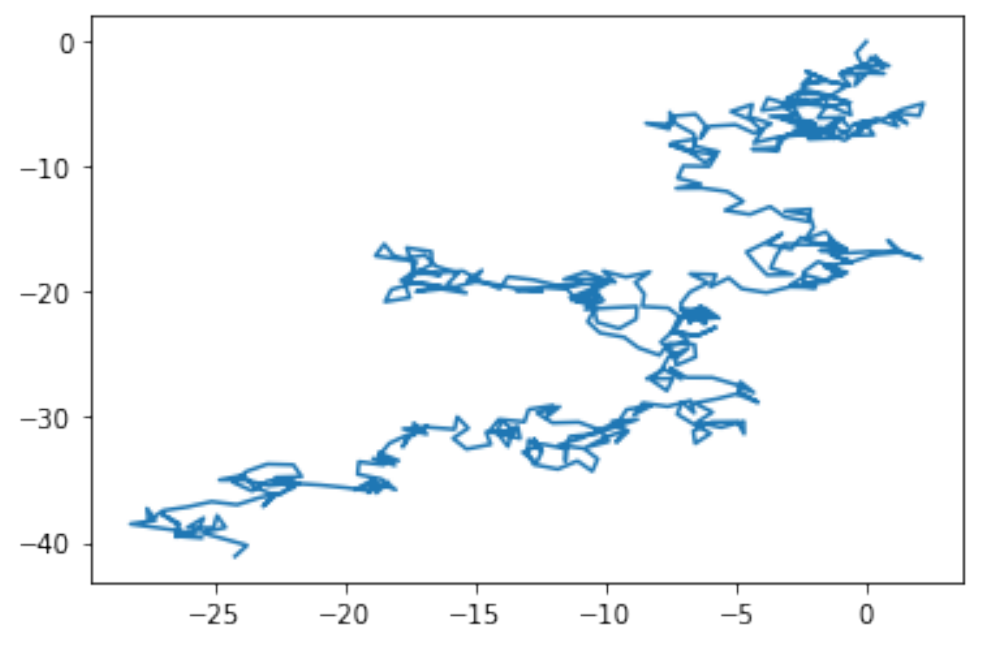}
         \caption{A fixed velocity random walk from $(0,0)$ to $(-24.21,-41.12)$.}
         \label{fig:path}
     \end{subfigure}
     \hfill
     \begin{subfigure}[b]{0.3\textwidth}
         \centering
         \includegraphics[width=\textwidth]{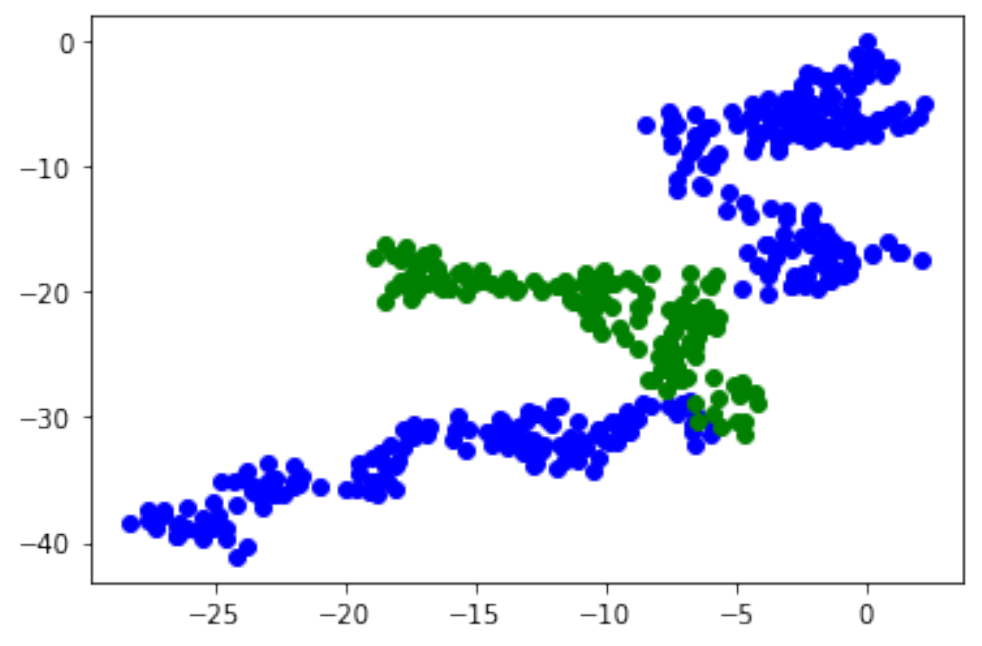}
         \caption{In green, a gap from time step $200$ to $350$.}
         \label{fig:gap}
     \end{subfigure}
     \hfill
     \begin{subfigure}[b]{0.3\textwidth}
         \centering
         \includegraphics[width=\textwidth]{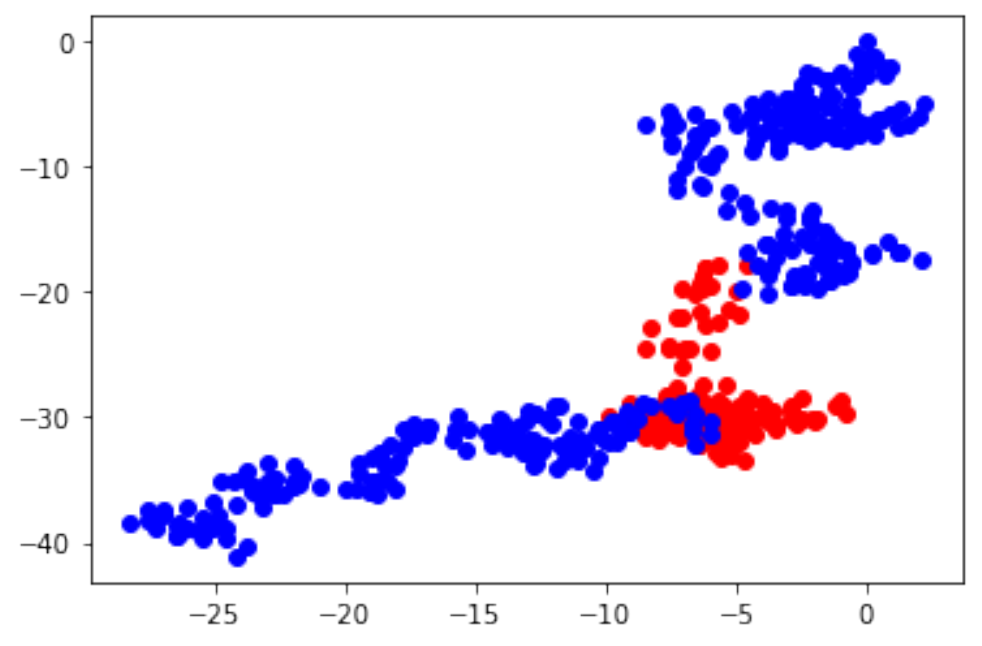}
         \caption{In red, the gap filled by a Brownian bridge with estimated $\sigma_m$.}
         \label{fig:fillgap}
     \end{subfigure}
        \caption{In \ref{fig:path}, a path is created of $500$ time steps which follows from the fixed velocity random walk model. Then, in \ref{fig:gap}, a gap is created where we omit the location points of all time steps in $[200, 350]$. Thereafter, in \ref{fig:fillgap} the gap is filled by creating a Brownian bridge of $150$ time steps with an estimated $\sigma_m$ from the incomplete path of \ref{fig:gap}.}
        \label{fig:procedure}
\end{figure}

With this setup we can measure how our estimation procedure compares to the actual value in terms of path length and RoG. Moreover, we will also compare our estimation procedure with linear interpolation. Linear interpolation fills the gap with a straight line with constant velocity. 

In the next two sections we will discuss in more detail which specific choices we have made regarding the production of data. We have used different models and different choices of parameters for the experiments. For the interested reader, we refer to Appendix \ref{app:processes} for the technical details of these data generating models. The reason behind this is that we had no access to empirical data. Note that comprehension of the generating models is not important for understanding our experiment.

\subsubsection{Experiments Regarding the Path Length}

For the experiments regarding path length, we have used four different models. These models are the discrete Brownian motion model, angular random walk model, random walk with internal state model and run-and-tumble model. Each model has its own parameters (for the details we refer again to Appendix \ref{app:processes}). These parameters affect the amount of randomness in their realisations. We will now describe our choices of parameters for reproducibility of the experiments. 

The discrete Brownian motion model has a clear parameter for the randomness, since it is made by realisations of a normal distribution with some variance. We will see that this variance $\sigma$ is indirectly controlled by the diffusion coefficient. In our experiments, we will use a fixed travelling distance of $10$ in $200$ time steps and model this with different standard deviations of $\sigma \in \{0.01,0.1,1,10\}$ to give a far spread in the amount of randomness in the system. As our estimation procedure is based on the assumption that the movement is a Brownian motion, we expect our procedure to perform well on this type of data. The result of this experiment should therefore be mainly seen as a sanity check that our procedure works for the best possible case.

The angular random walk model gives quite a natural way of moving in an open space without any abrupt stops or tight corners. This kind of movement can correspond to a car on a highway or a person walking in a field. Here the amount of randomness is dictated by the variance for the angular acceleration. For a very large variance we will almost get a completely uniform random direction in each time step (which is referred to as the fixed velocity random walk model), which will give more chaotic behaviour. We will consider standard deviations $\sigma \in \{0.1, 0.5, 1, 5\}$. The spread here is a bit smaller, since this model converges to straight lines quite fast for small $\sigma$, which is not that interesting to consider. 

The run-and-tumble model has more sharp corners and sudden changes in behaviour. This looks more of a movement in urban areas. The parameter $l$ here that we will vary is related to the chance that the direction will be changed. We will use the parameter $l \in \{ 0.1, 0.5, 1, 3\}$.

One of the similarities between the discrete Brownian motion-, angular random walk- and run-and-tumble model is that they all seem to behave the same at all times. In the random walk with internal state model, there is an internal state where it can occur that the movement stops for a while. Here there are discrete chances for what will be the next move. To dictate the randomness that will appear, we have varied the chances that the state changes and certain movements will happen. In the most extreme case, all options will be performed uniformly random.

For each such model, we have simulated paths of $200$ data points and created a large gap by removing $100$ data points in the middle of these paths. We have simulated this process $1000$ times for each choice of parameters. We have normalised each simulation with the actual path length. The results of the specific comparison with linear interpolation will be displayed in Figure \ref{fig:subfigures}.

\subsubsection{Experiments Regarding the RoG}

To be as concise as possible, we have considered three models (fixed velocity random walk, fixed angular random walk and run-and-tumble). The main idea of this experiment is to compare the RoG of a path with the RoG of the same path, but a part of the path is deleted and filled by a Brownian bridge (according to the estimation procedure). 

Each model has generated $1000$ paths where the corresponding parameter is fixed. Fixed velocity random walk with a velocity of $1$, fixed angular random walk with $\sigma =0.1$ and run-and-tumble model with $l=1$. We have chosen these settings, as they create paths with a small RoG and with a larger RoG. This way, we have more diversity in our results. For each path, the first half of the location points is deleted. The reason behind this is to consider an extreme case where the gap is rather large. The number of location points are $1000$, so the first $500$ points will be deleted. These are all the choices that we have made for these particular experiments. The results can be found in Section \ref{subsec:rog}.

\clearpage
\section{Introduction to Brownian Bridges} \label{sec:theory}

In the CBS travel survey, we sometimes encounter gaps in the data. We want to fill in these gaps by modelling a path between the last data point before and the first data point after the gap. One of the most common methods modelling this path is called the Brownian bridge. In this chapter we will set up a two-dimensional Brownian bridge and derive its distribution.

In the first section, we will elaborate on the mathematical theory of Brownian motion, and in the second section we construct a Brownian bridge. If the reader is already familiar with the concept of a Brownian bridge, then this chapter may be skipped.

\subsection{Brownian Motion}

Before we introduce Brownian motion, we will first explain the notion of a stochastic process in an informal manner, for more details and background one can look at \cite{kallenberg1997foundations}. To do so, let us consider the Euclidean space $\mathbb{R}^n$ and the movement of a particle in such a space. For example, such a particle could model a molecule in a solution or a person walking around. However, in most cases we cannot possibly predict such movement exactly. To capture such non-deterministic behaviour, we turn to probability theory and view the path itself as a collection of random vectors. 

In general, a \emph{stochastic process} in continuous time \cite[pg. 189]{klenke} is a collection $X = \{X_t: t \in I\}$ of random vectors in $\mathbb{R}^n$, where $I$ is some continuous index set, like $(0, \infty)$ or $[0, T]$. A continuous-time process $X$ is called \emph{continuous} if $t \mapsto X_t$ is continuous almost surely.
And in the case of movement through Euclidean space, $X_t$ represents the distribution of our paths at time $t$. 

Now, there exist many types of continuous stochastic processes. Although these processes are useful, as they can capture quite general dynamics, it can be difficult to explicitly calculate their distribution. So, we will only consider \emph{Brownian motion}. For a concrete definition, see \Cref{def:brownianmotion}. This is a stochastic process where all time increments are independent, normally distributed random variables. Specifically, let $W_t$ be a Brownian motion process. Then, for times $t$ and $s$ with $t \geqslant s$ we have that
\begin{equation*}
    W_t - W_s \sim W_{t-s} \sim N(0, t - s),
\end{equation*}
where $N(\mu, \sigma^2)$ denotes a normal distribution with mean $\mu$ and variance $\sigma^2$. Since we have such a concrete distribution for the intervals, it is very feasible to do calculations with Brownian motion.

One particular property we wish to highlight, is the scale invariance of the Brownian motion process. Notably, suppose we scale space by $\varepsilon$ and time by $\varepsilon^{-2}$, then we get the process $\varepsilon W_{\varepsilon^{-2} t}$. In particular, we see that for $t \geqslant s$
\begin{equation*}
    \varepsilon W_{\varepsilon^{-2} t} - \varepsilon W_{\varepsilon^{-2} s} \sim \varepsilon N(0, \varepsilon^{-2}(t -s)) \sim N(0, t - s).
\end{equation*}
As the above in a sense entirely characterises the Brownian motion, we see that $\varepsilon W_{\varepsilon^{-2} t} \sim W_t$. This scale invariance means that the movement of the process looks the same no matter at what scale one considers the process. In general, this scale invariance is not observed in the movement of people. On a scale of centimetres, a person will walk in a straight line. On a scale of hundreds of metres, they will not move in a straight line due to constraints of the local environment i.e. how one is forced to move due to obstructions such as buildings. However, on a scale of kilometres, they will move roughly in a straight line, since one generally wants to take the shortest path between two locations. Hence, when modelling any type of movement, one should account for the scale of the movement being modelled. We will see this again in Chapter \ref{sec:simulation}.

\subsection{Brownian Bridge}\label{subsec:brownianbridge}

Now let us move on to a formal construction of the Brownian bridge. This section is inspired by \cite{animalBBMM}. Let $\{W_t\}_{t \geqslant 0}$ be a Brownian motion and fix a time $T \in (0, \infty)$. Then we can define a \emph{standard one-dimensional Brownian bridge} \cite[pg. 175]{shrevenfinancialmath} by $X = \{X_t: t \in [0, T]\}$, with 
\begin{equation*}
X_t := W_t - \frac{t}{T} W_T.
\end{equation*}
This is a standard one-dimensional Brownian bridge in the sense that one can interpret it as a continuous walk on the real line that begins and ends in zero.
We can also write $X_t$ as
\begin{equation*}
    X_t = \frac{T-t}{T} W_t - \frac{t}{T} (W_T - W_t).
\end{equation*}
Since $W$ is a Brownian motion, it follows that $W_T - W_t \sim W_{T-t}$. This implies that
\begin{equation*}
    \frac{T-t}{T} W_t \sim N\left(0, \frac{t(T-t)^2}{T^2}\right)\;\;\; \text{ and }\;\;\; \frac{t}{T}(W_T - W_t) \sim N\left(0, \frac{t^2(T-t)}{T^2}\right).
\end{equation*}
A Brownian motion has independent increments, so
\begin{equation*}
    X_t \sim N\left( 0, \frac{t(T-t)^2}{T^2} + \frac{t^2(T-t)}{T^2}\right) \sim N\left(0, \frac{t(T-t)}{T}\right)
\end{equation*}
as the sum of independent normals is again normal, see \Cref{prop:sumnormals}. 

Note that the process is still purely one-dimensional. The application to location data is of course two-dimensional, so we need to extend this idea. Since we expect there to be no preference for a specific cardinal direction, we may assume that the coordinates will be independent. That way, if we let $Y$ be another standard Brownian bridge distributed identically to $X$, then the standard Brownian bridge in two dimensions is simply
\begin{equation*}
    \begin{pmatrix} X_t \\ Y_t \end{pmatrix} \sim N_2\left(0, \frac{t(T-t)}{T} \I_2 \right),
\end{equation*}
as a vector of two i.i.d.\ normal random variables becomes a bivariate normal random vector, by \Cref{prop:distrnormals}.

When we look at this process, we see that when $t = 0$ and when $t = T$, the process has an $N_2(0, 0\I_2)$ distribution, which means it is a constant in the origin $O$, i.e.\ the process will always begin and end in $O$.
\begin{figure}[H]
    \centering
    \includegraphics[width=0.7\linewidth]{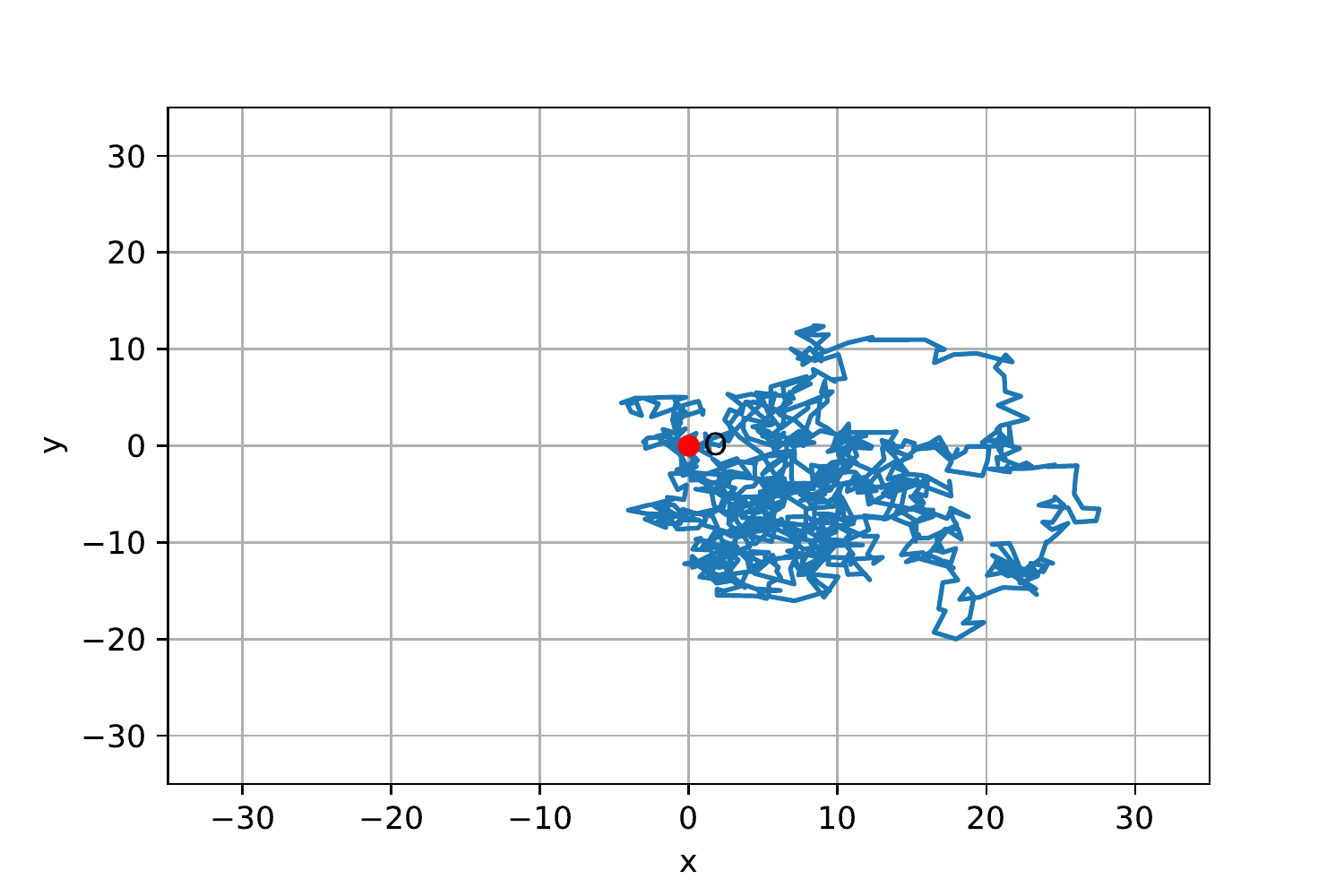}
    \caption{A simulated realisation of a two-dimensional Brownian bridge starting and ending in $O$, with $T = 1000$.}
    \label{fig:standard2DBB}
\end{figure}
As we can see in the figure above, this process simulates a \emph{round trip}: we begin and end in the origin. However, we will be using this process to fill in a gap between two different points. We can do this easily by adding a time dependent path between these points. Suppose we want to model a path from the origin $O$ to some $d \in \R^2$, then on average, we assume the path to be a straight line. This is parametrised by the function 
\begin{equation}
    \mu(t) = \frac{t}{T} d.
\end{equation}
This is a two dimensional vector, so we can simply say that our process becomes $\mu(t) + (X_t, Y_t)$. Recall here that we have made the assumption that the traveller travels in a straight line. With no further information about the traveller or their surroundings, this is a valid assumption. If one has reason to believe the traveller would behave differently, one can model that with a different function for $\mu$.

Since there are many different ways of travel, we want to include another parameter that accounts for this. In the literature it is often called the \emph{diffusion coefficient}, and we denote it by $\sigma_m$. It works as follows. If we consider $\mu(t) + \sigma_m(X_t, Y_t)$, then this is a normally distributed random vector with covariance matrix
\begin{equation*}
    \sigma_m \frac{t(T-t)}{T} \I_2.
\end{equation*}
We see that if $\sigma_m$ is small, the process is unlikely to move far from its expected value, in this case $\mu(t)$. If $\sigma_m$ is large, it becomes quite likely to do so. In this way we can model different modes of transport: a train is much more likely to move in a straight line than someone on a bicycle, so the train will have a lower $\sigma_m$ than the bicycle.

With this, we can finally define the process we will use to model what happens in a gap in the data. We define a \emph{general Brownian bridge} $Z = \{Z_t: t \in [0, T]\}$ by
\begin{equation*}
    Z_t := \mu(t) + \sigma_m \begin{pmatrix} X_t \\ Y_t \end{pmatrix}.
\end{equation*}
If we let $\sigma^2(t) := \sigma_m^2 t (T-t) / T$, then $Z_t \sim N_2(\mu(t), \sigma^2(t)\I_2)$. Let us look at a simulated realisation of such a process.
\begin{figure}[H]
    \centering
    \includegraphics[width=0.7\linewidth]{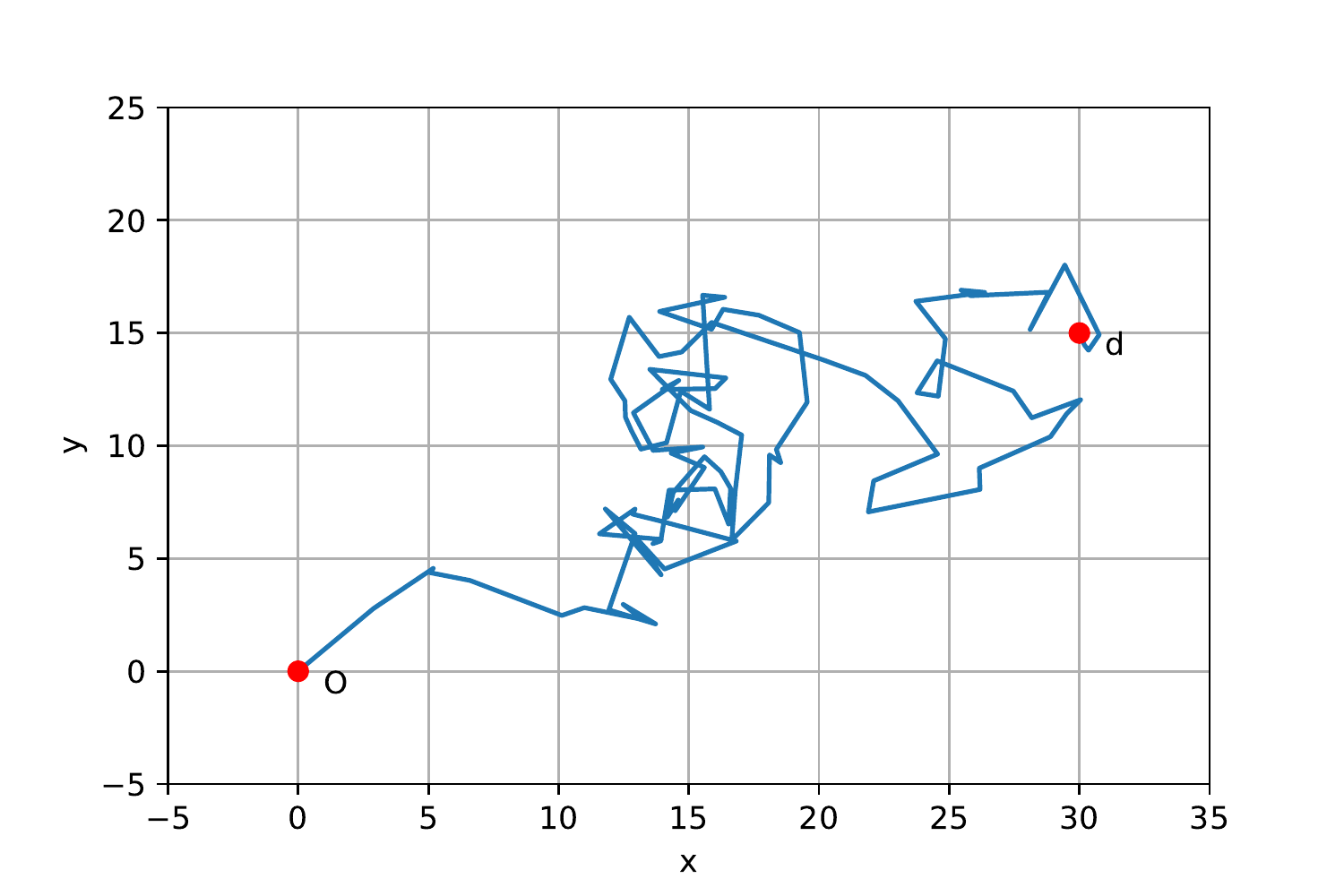}
    \caption{A simulated realisation of a general Brownian bridge with $\sigma_m = 2$, $T = 100$ and $d = (30, 15)$.}
    \label{fig:general2DBB}
\end{figure}
We see that this process indeed simulates a path from $O$ to $d$. As is usual for a Brownian movement, the path is fairly erratic. However, on average it doesn't stray too far from the path. This is because $\sigma_m$ is quite small.

\clearpage
\section{Model} \label{sec:model}

We are interested in modelling gaps in the location data of a traveller. A common way to model movement of a person or animal is with the Brownian Bridge Movement Model, or BBMM \cite{animalBBMM}. This can be used to fill the gaps in between each measurement, or to fill in a larger gap, where data is missing. We will be doing the latter. Recall that we are not interested in the exact path a traveller takes, as this can be very different with each realisation. Furthermore, to be realistic it should depend on the traveller's environment, which is something that we cannot take into consideration, as it would make the problem much too complex. Instead, we are interested in more general properties, like the expected length of the travellers path and the RoG. We will define these concepts in the following sections.

\subsection{Expectation of Discrete Brownian Bridge Path Length}\label{subsec:pathlength}

In this section we will derive an expression for the expectation of the length of a realisation of a discretised Brownian bridge. To do this, we first need to discretise the continuous time Brownian bridge we have already seen, and then we define a path using that discrete process. We will investigate the distribution of the length of that path, and look at some limiting properties of its expected value.

\subsubsection{Path Length of a Discretised Brownian Bridge}

The paths that are defined in Section \ref{subsec:brownianbridge} are all paths continuous over time. We discretise the general Brownian bridge process as follows.

Consider the process $Z = \{Z_t : t \in [0, T]\}$ with 
\begin{equation} \label{eq:generalbridge}
    Z_t = \mu(t) + \sigma_m \begin{pmatrix} X_t \\ Y_t \end{pmatrix} \sim N_2\left(\mu(t), \sigma^2(t) \I_{2}\right),
\end{equation}
where $\mu(t) = td/T$ and $\sigma^2(t) = \sigma_m^2t(T-t) / T$, as defined in Section \ref{subsec:brownianbridge}. Fix $n \in \N$ and let $0 = t_0 < t_1 < \dots < t_{n-1} < t_n = T$ be any discretisation of the interval $[0, T]$. A path can now be defined as taking straight lines between these $n+1$ points. That is, define the path
\begin{equation*}
    t \mapsto \sum_{k =1}^n \one_{\left(t_{k-1}, t_k \right]}(t) \left(Z_{t_{k-1}} + (Z_{t_k} - Z_{t_{k-1}}) (t - t_{k-1}) \right).
\end{equation*}
The length of this path is then given by 
\begin{equation*}
    \ell^n(Z) = \sum_{k=1}^n \norm{Z_{t_k} - Z_{t_{k-1}}}. 
\end{equation*}
This path length is a random variable, and though it is a function of a process with a fairly simple distribution, the square root makes the situation a bit more complicated. In the next part we will find the distribution of this path length.

\subsubsection{Distribution and Results} \label{subsec:bbresults}

To find an explicit distribution for $\ell^n(Z)$, we need to look at the distribution of $Z_{t_k} - Z_{t_{k-1}}$. Recall from \eqref{eq:generalbridge} that $Z_t$ has a normal distribution with parameters $\mu(t)$ and $\sigma^2(t)$. Then
\begin{equation*}
    Z_{t_k} - Z_{t_{k-1}} = \frac{t_k - t_{k-1}}{T}d + \sigma_m \begin{pmatrix} X_{t_k} - X_{t_{k-1}} \\ Y_{t_k} - Y_{t_{k-1}} \end{pmatrix},
\end{equation*}
so we first need to find the distributions of $X_{t_k} - X_{t_{k-1}}$ and $Y_{t_k} - Y_{t_{k-1}}$. This follows fairly straightforwardly from one of the properties of the Brownian motion, as seen in the proof of \Cref{prop:1dbridgedifference}. This proposition tells us that $X_{t_k} - X_{t_{k-1}} \sim X_{t_k - t_{k-1}}$ and $Y_{t_k} - Y_{t_{k-1}} \sim Y_{t_k - t_{k-1}}$, so
\begin{equation*}
    Z_{t_k} - Z_{t_{k-1}} \sim \mu(t_k - t_{k-1}) + \sigma_m \begin{pmatrix} X_{t_k - t_{k-1}} \\ Y_{t_k - t_{k-1}} \end{pmatrix} = Z_{t_k - t_{k-1}}.
\end{equation*}
This implies that we can write
\begin{equation*}
    \ell^n(Z) \sim \sum_{k=1}^n \norm{Z_{t_k - t_{k-1}}}.
\end{equation*}
This holds for any discretisation $0 = t_0 < \dots < t_n = T$ of the interval $[0, T]$. A reasonable assumption would be to split $[0, T]$ into $n$ intervals of equal length by setting $t_k = kT/n$. Then, $Z_{t_k - t_{k-1}} = Z_{T/n}$ for each $k \in \{1, \dots, n\}$. That means that
\begin{equation} \label{eq:discpathlength}
    \ell^n(Z) \sim \sum_{k=1}^n \norm{Z_{T/n}} = n \norm{Z_{T/n}} = \norm{nZ_{T/n}}.
\end{equation}
As $\mu(T/n) = d/n$ and $\sigma^2(T/n) = \sigma_m^2T(n-1) / n^2$, we see that $nZ_{T/n}$ is distributed as
\begin{equation*}
    nZ_{T/n} \sim N_2\left( d, \sigma_m^2T(n-1)\I_2 \right).
\end{equation*}
\Cref{def:riciandistribution} now tells us that $\ell^n(Z) \sim \norm{nZ_{T/n}}$ has the Rice distribution with parameters
\begin{equation*}
    A = \norm{d}, \;\;\; B = \sigma_m^2T(n-1).
\end{equation*}
By that same definition we find that the expected value of the path length $\ell^n(Z)$ is given by
\begin{equation}\label{eq:pathlengthexpectation}
    \E[\ell^n(Z)] = \sigma_m\sqrt{T(n-1)}\sqrt{\frac{\pi}{2}}L_{\frac{1}{2}}\left(-\frac{\norm{d}^{2}}{2\sigma_{m}^{2}T(n-1)}\right). 
\end{equation}
Note here that $L_{\frac12}$ is a Laguerre function \cite{laguerre}. This expected value is the main result of this section. In the next part, we will look at limits of this expectation, but first we will see that this expected value has a lower bound of $\norm{d}$.
This is quite a logical lower bound: you cannot find a shorter distance than the shortest distance, the norm. Applying Jensen's inequality \cite[pg. 107]{IKSboek}
to the convex function $(x, y) \mapsto \norm{(x, y)}$ we indeed see that
\begin{equation*}
    \E \left[\ell^n(Z)\right] = \E \norm{nZ_{T/n}} = n\E \norm{Z_{T/n}} \geqslant n \norm{\E\left[Z_{T/n}\right]} = \norm{d}.
\end{equation*}

\subsubsection{Limiting Properties of the Expected Path Length}\label{subsec:limits}

Let us look at some limits of the expected path length, i.e.\ the function
\begin{equation*}
    f(\sigma_m, T, d, n) := \E[\ell^n(Z)] = \sigma_m \sqrt{T(n-1)}\sqrt{\frac{\pi}2} L_{\frac12}\left(-\frac{\norm{d}^2}{2\sigma_m^2 T(n-1)} \right),
\end{equation*}
parameter by parameter. We will use the following property of the Laguerre function to calculate limits of $f$,
\begin{equation} \label{eq:limit}
\lim_{x \to \infty} \sqrt{\frac{\pi}{2x}} L_{\frac12}\left( - \frac12 x \right) = 1.
\end{equation}
Although it is possible to show this analytically using Bessel functions, we have simply plugged this expression into a graphing software.
First, let us look at limits of $\sigma_m$. If $\sigma_m \downarrow 0$, we know for certain that the traveller moves over the straight line from origin $O$ to endpoint $d$, without deviating from that path. However, if $\sigma_m$ is large, the traveller is very likely to deviate from the straight line. So likely in fact, that the length of the path they walk blows up, as they might go very far from the straight line before they finally end up in $d$. Using \eqref{eq:limit}, we indeed see that
\begin{equation*}
    \lim_{\sigma_m \downarrow 0} f(\sigma_m, T, d, n) = \norm{d}, \;\;\; \lim_{\sigma_m \to \infty} f(\sigma_m, T, d, n) = \infty.
\end{equation*}

For $T$ the intuition is a bit different. We get the limits
\begin{equation*}
    \lim_{T \downarrow 0} f(\sigma_m, T, d, n) = \norm{d}, \;\;\; \lim_{T \to \infty} f(\sigma_m, T, d, n) = \infty,
\end{equation*}
but if $T = 0$, our problem is undefined whenever $d \neq O$. When $d = O$, the Brownian bridge process is just a constant, namely $O$, so then the limit is reasonable, but useless. The limit $T \to \infty$ makes a bit more sense. Then the traveller wanders around so long, they will never arrive, which means that their path length must be infinite as well. Note that although $n$ is discrete, we might consider an analytic continuation. In this case, the behaviour of $n - 1$ is exactly identical to that of $T$. The limit $n \to 1$ then corresponds to estimating the path by a single line segment for which we get an expectation of $\norm{d}$. On the other hand, if $n$ goes to infinity, the path length goes to infinity. This is as expected, since as $n$ goes to infinity, the random walk becomes Brownian motion, which has an infinite path length. 

Moving on to limits of $d$, it follows that since $\norm{d}$ is a lower bound for the expected path length, the expected path length blows up as $\norm{d} \to \infty$. The limit where $\norm{d} \downarrow 0$, i.e.\ $d \to O$, is a round trip. We find
\begin{equation*}
    \lim_{\norm{d} \to \infty} f(\sigma_m, T, d, n) = \infty, \;\;\; \lim_{\norm{d} \downarrow 0} f(\sigma_m, T, d, n) = \sigma_m \sqrt{\frac\pi2 T(n-1)}.
\end{equation*}
Recall that $\ell^n(Z)$ has the Rice distribution. It turns out that when $\mu(t) = 0$, as it is when $d = O$, the distribution of $\ell^n(Z)$ is a special case of the Rice distribution, called the Rayleigh distribution, see \Cref{def:rayleighdistribution}. This distribution indeed has the expected value given above.

We see that all limits of this function are as we would expect them to be, so \eqref{eq:pathlengthexpectation} does seem to make sense. Since $d$ and $T$ follow directly from the data, and the choice of $n$ is up to us, the only thing we need to consider is how to choose $\sigma_m$. It turns out that we can estimate a value for $\sigma_m$ based on the data points we do have, i.e.\ the data points before and/or after the gap.

\subsection{Estimating the Diffusion Coefficient}

It is essential to estimate $\sigma_m$ correctly since the modelled path depends heavily on this. We expect that this diffusion coefficient can be very different when considering two distinct travel movements. For example, a train moves in a very straight line compared to a pedestrian walking in a city centre and thus their diffusion coefficients should reflect that. Note that we can not determine a uniform diffusion coefficient for a fixed mode of transportation, there are many more parameters that have an influence on the value of $\sigma_m$. For instance, a car driving on a highway is expected to have a different diffusion coefficient than a car driving in the centre of a city, as a car on a highway has a very linear path and in a city centre there are a lot of opportunities to turn.

Hence, we want to estimate $\sigma_m$ for some travel movement based on the available location data for that movement. To do this, we consider location data in a single movement, with a single mode of transportation. The following is based on the section \emph{Parameter estimation} in \cite{animalBBMM}.

Let us assume we have $2n+1$ location measurements, denoted by $(z_0, t_0), (z_1, t_1), \dots, (z_{2n}, t_{2n})$, where $z \in \R^2$ is the travellers location at time $t$. On each time interval $[t_{2k}, t_{2k+2}]$ we will construct a Brownian bridge with diffusion coefficient $\sigma_m$ between $z_{2k}$ and $z_{2k+2}$, and then we will view $z_{2k+1}$ as a realisation of that Brownian bridge at time $t_{2k+1}$. Doing this for every increment, we can find a maximum likelihood estimator for $\sigma_m$.

Let us look at the distribution of the random variable of which we assume that $z_{2k+1}$ is a realisation. Define 
\begin{align*}
    T_k = t_{2k+2} - t_{2k}, 
    \qquad d_k = z_{2k+2} - z_{2k} 
    \qquad m_k(t) = \frac{t}{T_k}d_k, 
    \qquad s_k^2(t) = \sigma_m^{2} \frac{t(T_k - t)}{T_k}
\end{align*}
for $t \in [0, T_k]$.\footnote{Note that compared to the discussion in \cite{animalBBMM}, we ensure our time interval for each Brownian bridge $k$ starts in 0. Furthermore, in \cite{animalBBMM} there is a parameter $\delta$ which represents the vector of the locations error in the data. We assume that our data is perfectly correct.} Now the Brownian bridge between $z_{2k}$ and $z_{2k+2}$ looks like $z_{2k} + Z^k$, where $Z^k = \{Z_{t - t_{2k}}^k: t \in [t_{2k}, t_{2k+2}]\}$ with
\begin{equation*}
    Z_t^k \sim N_2(m_k(t), s_k^2(t)\I_2).
\end{equation*}
The observation $z_{2k+1}$ is then a realisation of the random vector $z_{2k} + Z_{t_{2k+1} - t_{2k}}^k$. This random vector has the probability density function
\begin{equation*}
    f_k(z) = \frac{1}{2\pi s_k^2(t_{2k+1} - t_{2k})} \exp \left( -\frac12 \frac{\norm{z - z_{2k} - m_k(t_{2k+1} - t_{2k})}^2}{s_k^2(t_{2k+1} - t_{2k})} \right),
\end{equation*}
see \Cref{def:pdf2}. We assume that all these bridges are independent, so the likelihood function of $\sigma_m$ given the realisations $(z_0, t_0), \dots, (z_{2n}, t_{2n})$ is 
\begin{equation}\label{eq:likelihood}
    L(\sigma_m \mid (z_0, t_0), \dots, (z_{2n}, t_{2n})) = \prod_{k=0}^{n} f_k(z_{2k+1})
\end{equation}
As all parameters in this equation are known except for $\sigma_m^2$, we can use the maximum likelihood estimator to estimate it. To do this we can numerically optimise the likelihood function over values of $\sigma_m^2$. By using the known data points we can estimate $\sigma_m^2$ such that the probability that the odd locations came from the Brownian bridges between the even locations is maximal.

Now that we have estimated $\sigma_m$, we can estimate the path length through equation \ref{eq:pathlengthexpectation}. Here, we have currently fixed $n$ to be number of missing data points. However, this might not be the best choice, as a choice of discretisation is essentially a choice of minimum straight path length. Increasing $n$ will break up a straight segment into a few smaller segments increasing the path length. So, for different types of processes, we might tune the number of points in our discretisation to more accurately model the path length behaviour of the considered process. One might use data for which the path length is known to get a estimate of what $n$ gives a good estimate of the path length. We will not consider this idea further because it goes beyond the scope of this report.

\subsection{Radius of Gyration} \label{subsec:rogtheory}

Consider a dataset of the form $(z_0, t_0), \dots, (z_n, t_n)$, so at time $t_i$ we observe the location $z_i$.We call this a path $P$.\footnote{This notation will be useful later.} Considering the fact that our location data is given as coordinates in a plane, we can model the RoG of $P$ as
\begin{equation}\label{eq:RoG}
    \text{RoG} = \sqrt{\frac{1}{n}\sum_{i=1}^n \norm{z - z_c}^2},
\end{equation}
where $z_c = \bar{z}_n = \frac1n \sum_{i=1}^n z_i$ can be seen as a weighted centre, akin to a centre of gravity.
This definition for the RoG originates from equation (4) in  \cite{9378374}. Note that we have simplified this formula due to our assumption that the location points are coordinates in a plane.
\begin{figure}[H]
    \centering
    \includegraphics[width=0.5\textwidth]{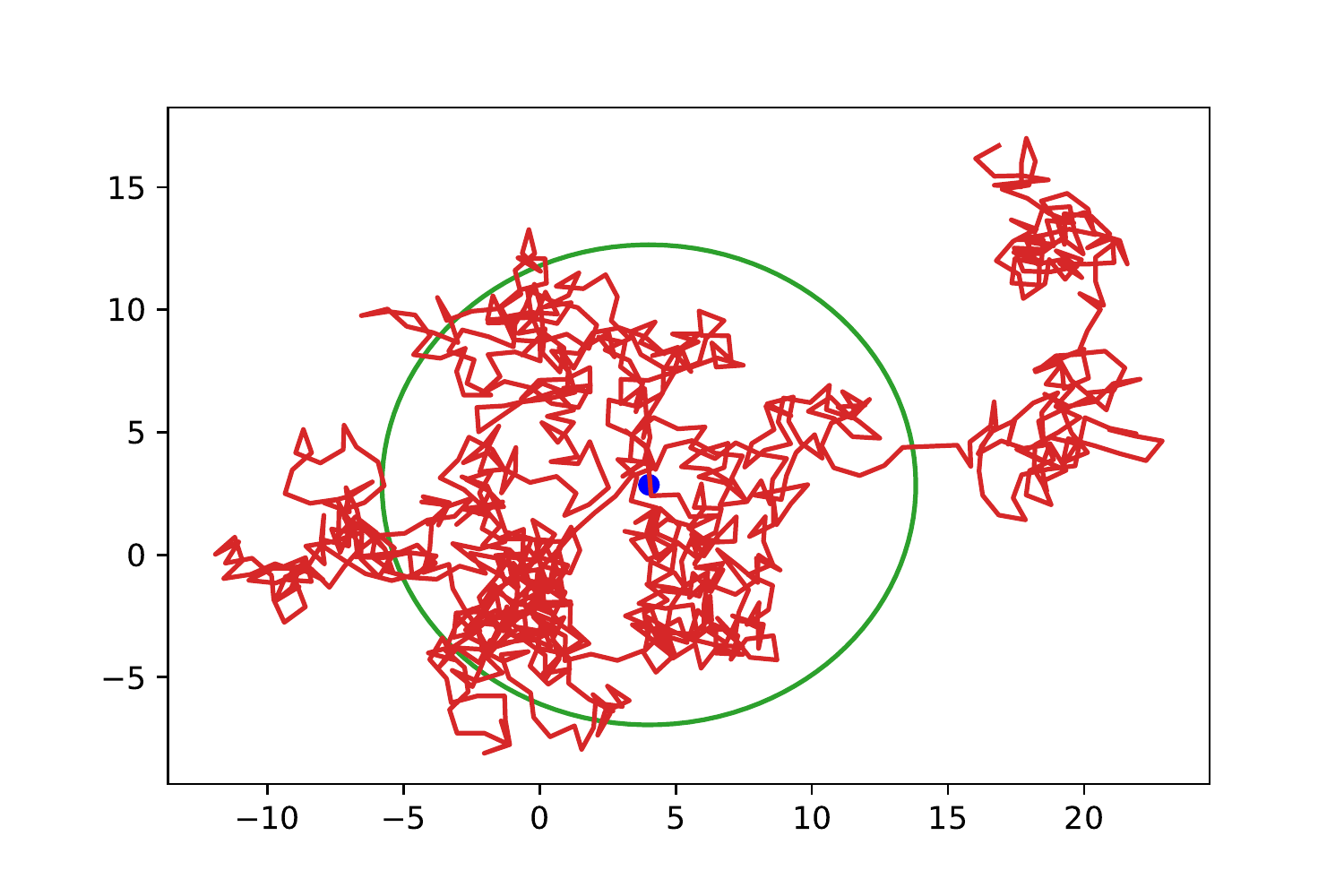}
    \caption{A path generated by the fixed velocity random walk of $1000$ time steps with $\sigma =1$. A circle is plotted with the RoG as a radius. The blue dot is the centre. This circle can be seen as the ``territory'' of the person.}
    \label{fig:rog}
\end{figure}

\clearpage
\section{Results of the Simulation}\label{sec:simulation}

In Section \ref{subsec:implementation} we have discussed the implementation of the Brownian bridge model. After that, we described the setup of our experiments. 
One of the topics that we have discussed is the filling of a gap by a Brownian bridge. A question regarding this topic could for example be: \textit{If we use a Brownian bridge with an estimated diffusion coefficient to fill a gap, is the path length then similar to the original part of the path that is missing?} 
In this section, we will show the results of the experiments. The main idea of these results is to see if our estimator for the path length works compared to linear interpolator (draw a straight line that connects the two sub paths divided by the gap). 

Also recall that in Section \ref{subsec:implementation} and Chapter \ref{sec:theory}, we have shortly described a procedure to estimate the diffusion coefficient assuming that the path follows a Brownian bridge. To determine the correctness of our estimator, we have simulated a thousand Brownian bridges with different variances. For each simulated path, the estimation procedure is used to obtain the most likely estimate for $\sigma_m$. Thereafter, we have used $\sigma_m$ to derive the actual variance $\sigma$. In Figure \ref{fig:graph2}, a plot of the actual variance versus this estimated variance is displayed.
\begin{figure}[h!]
    \centering
    \includegraphics[width=0.6\textwidth]{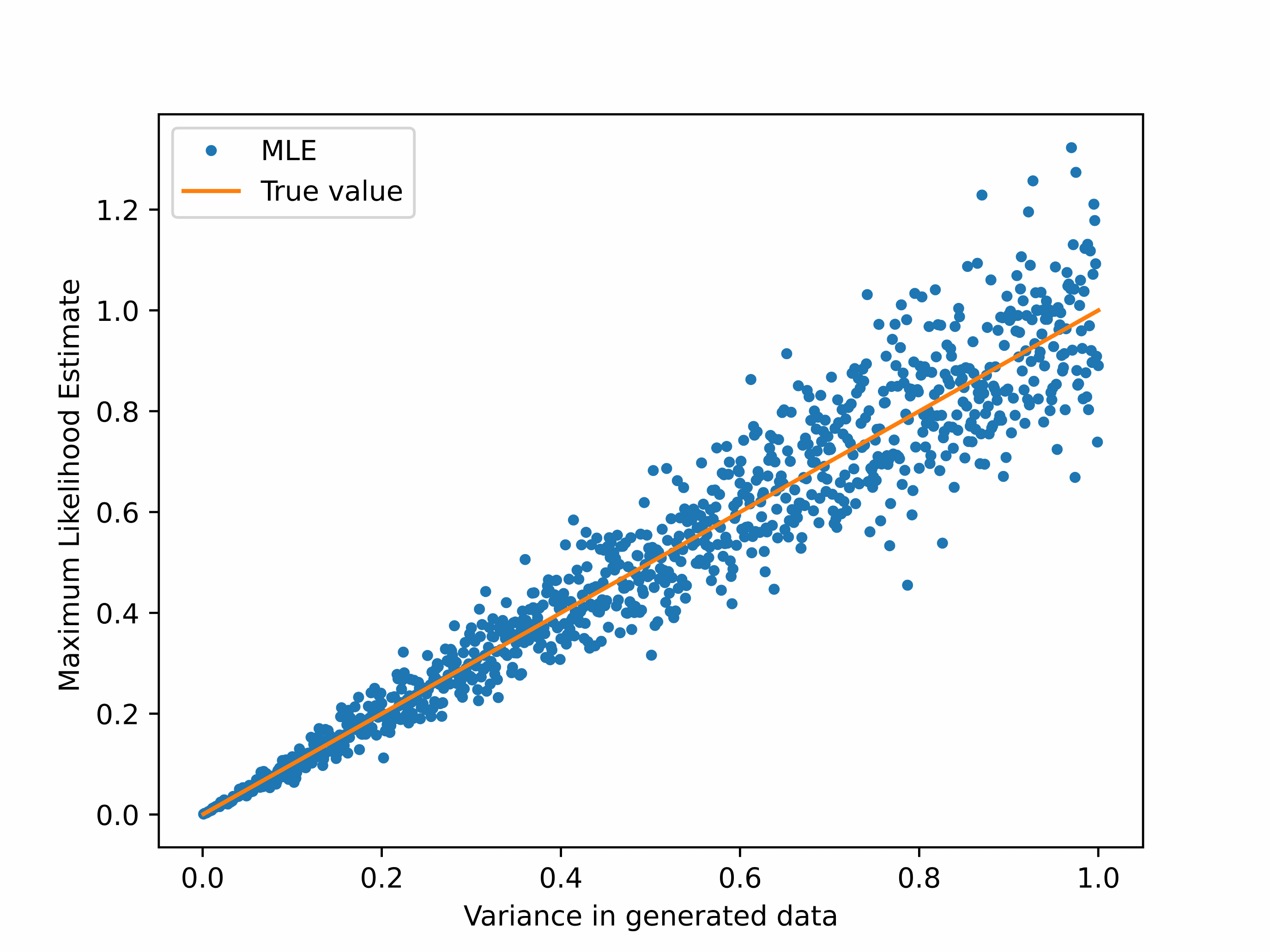}
    \caption{The actual variance $\sigma$ ranges from $0$ to $1$. On the $x$-axis, the actual variance is given. On the $y$-axis, the estimated variance is given. The orange line is drawn as a reference to the perfect estimation or ``true value''. So, if the actual variance is $1$, on the orange line the estimated variance is also $1$. For each actual variance, the blue dots are the estimated variances.}
    \label{fig:graph2}
\end{figure}
We see that our estimation procedure comes quite close to the actual diffusion coefficient. As the diffusion coefficient becomes greater the absolute error seems to increase. 
\subsection{Path Length}
In Section \ref{subsec:experiments} we have fully described our experiment to quantify the accuracy of our estimation procedure. We have performed those simulations and got the results depicted in \Cref{fig:subfigures}. The error of our estimation is shown as a fraction of the actual path length. In particular the estimated path length divided by the actual path length is shown, this means an error of $1$ is a perfect estimation.

In Figure \ref{fig:subfigures}, it can be seen how our estimation procedure and the straight line estimate performed on different kind of datasets. For all described choices of parameters, a box plot is shown with the estimated path length as a fraction of the actual path length. 
\begin{figure}[H] 
    \centering
	\begin{subfigure}[b]
		{0.49\textwidth}
		\centering
		\includegraphics[width=\linewidth]{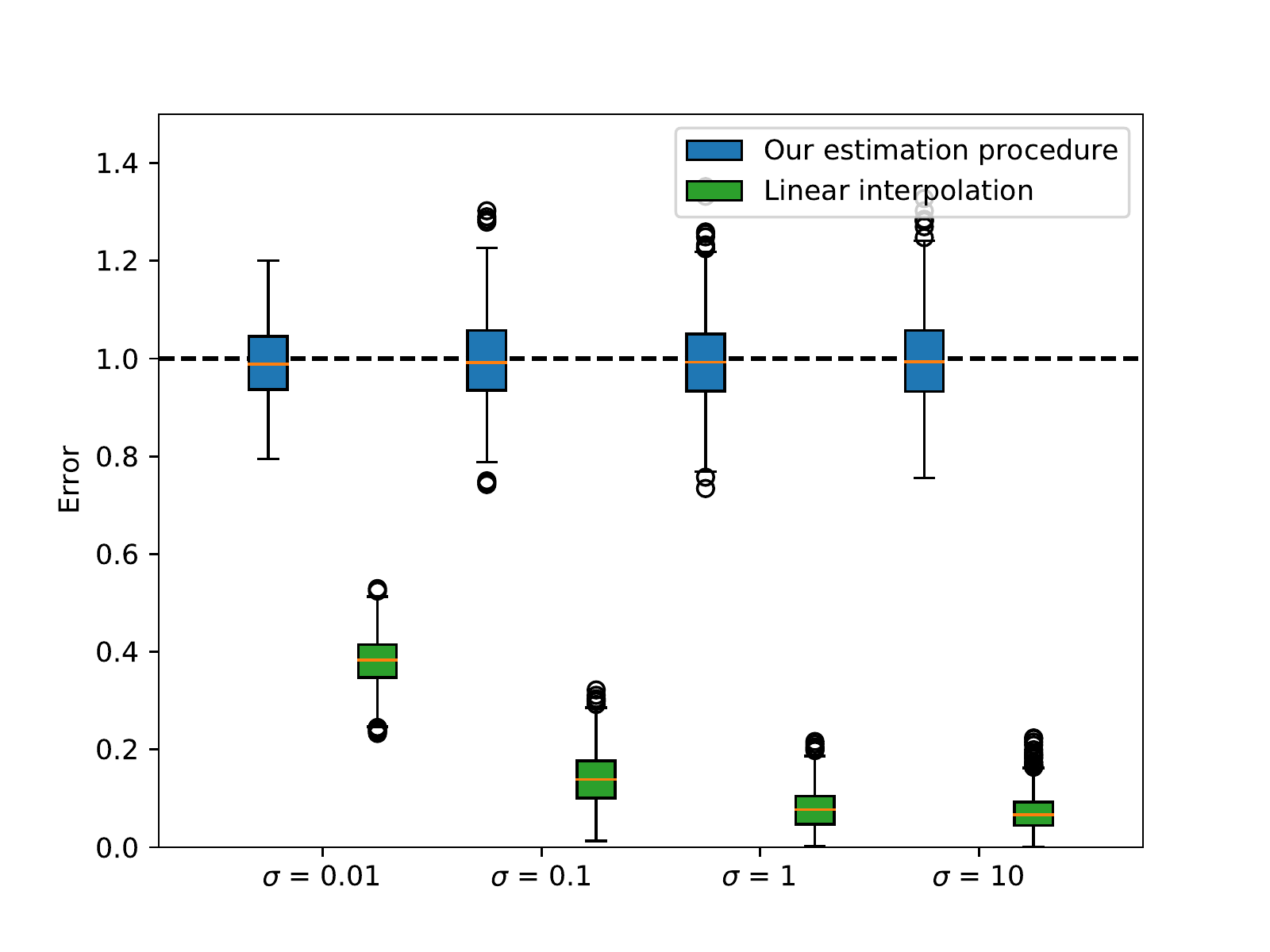}
		\caption{Results for discrete Brownian motion model}
		\label{fig:brownianmotion}
	\end{subfigure}
	\begin{subfigure}[b]
		{0.49\textwidth}
		\centering
		\includegraphics[width=\linewidth]{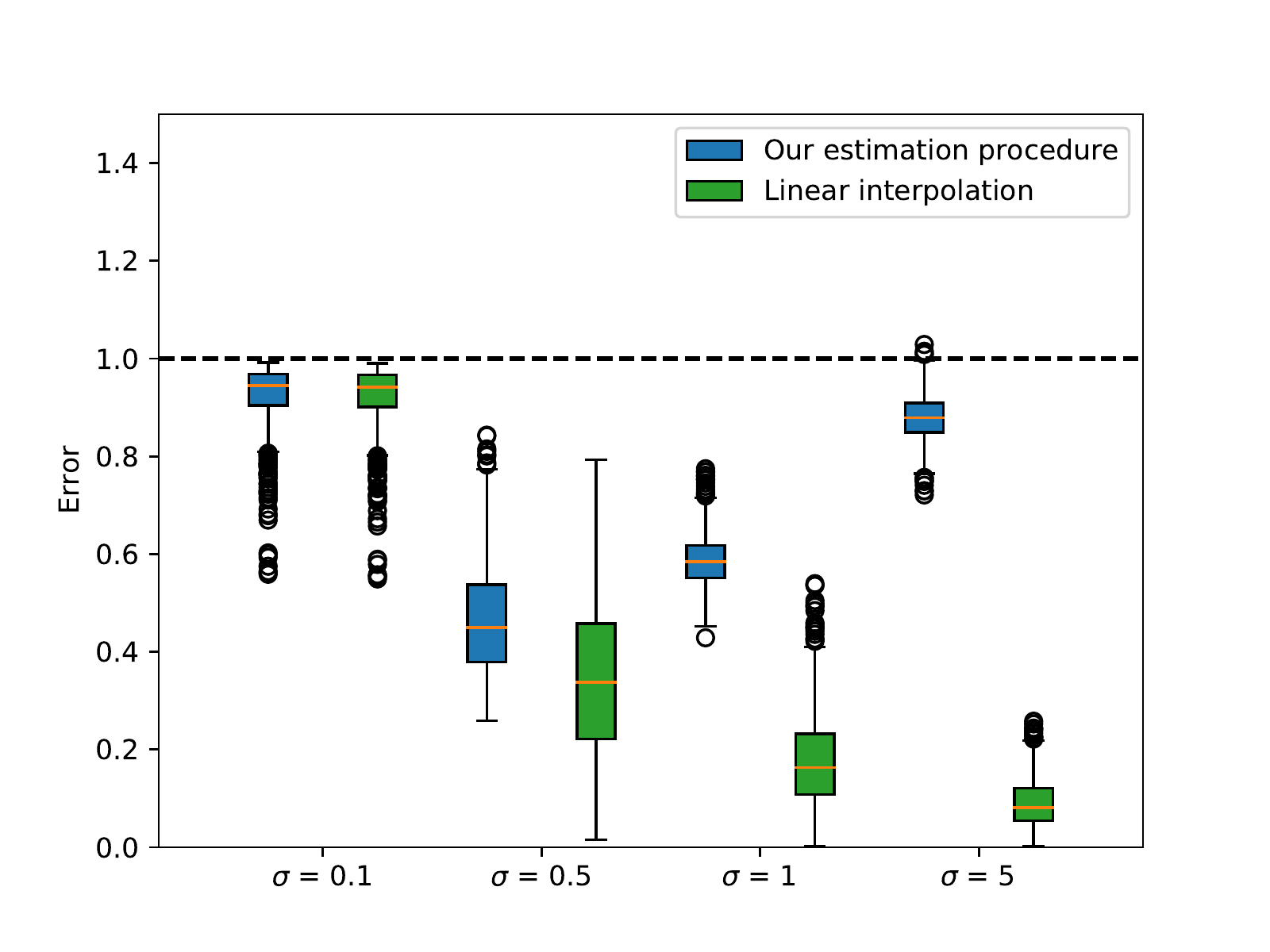}
		\caption{Results for angular random walk model}
		\label{fig:angularrandomwalk}
	\end{subfigure}
	\begin{subfigure}[b]
		{0.49\textwidth}
		\centering
		\includegraphics[width=\linewidth]{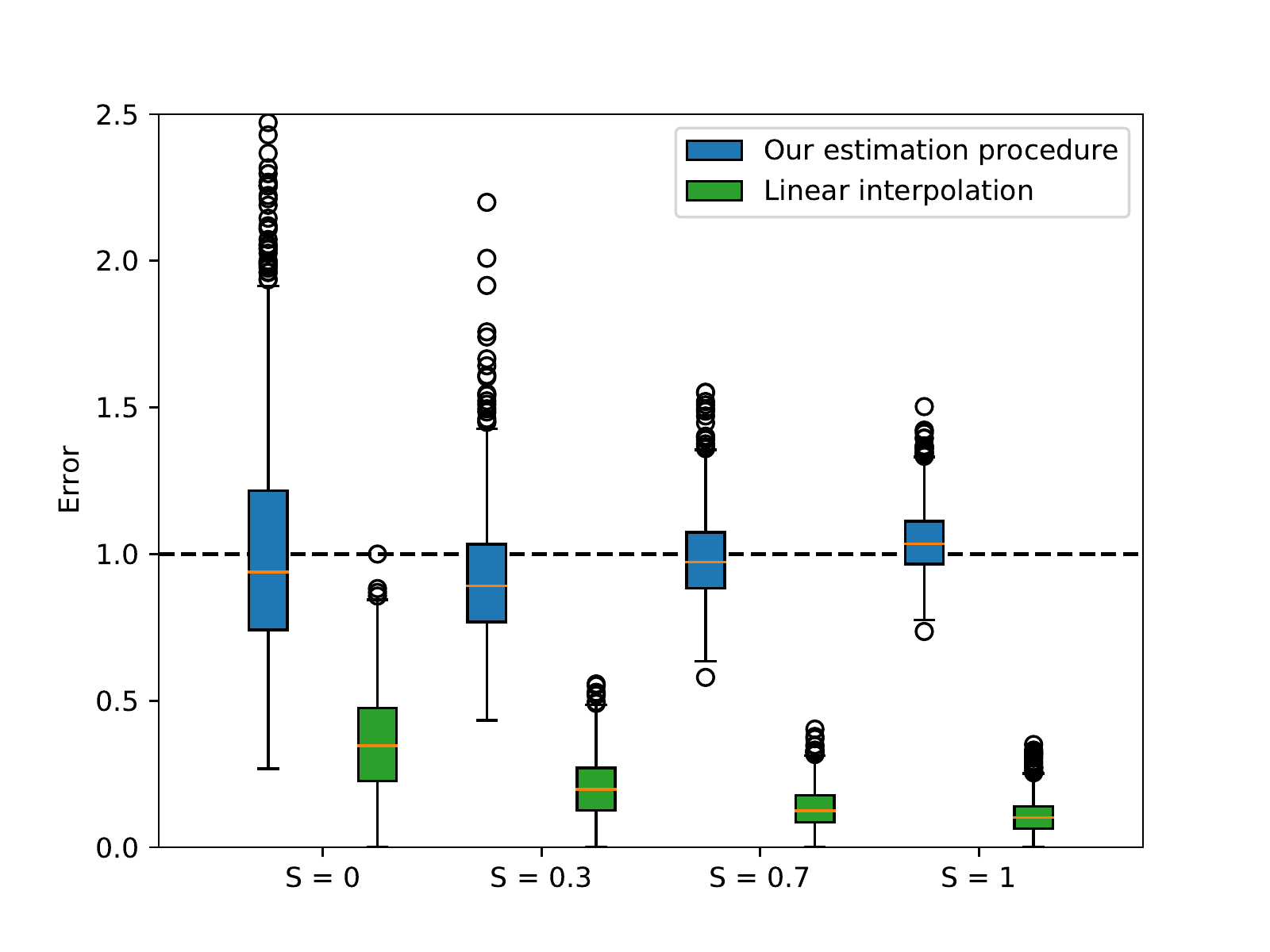}
		\caption{Results for random walk with internal state model.}
		\label{fig:internalstate}
	\end{subfigure}
	\begin{subfigure}[b]
		{0.49\textwidth}
		\centering
		\includegraphics[width=\linewidth]{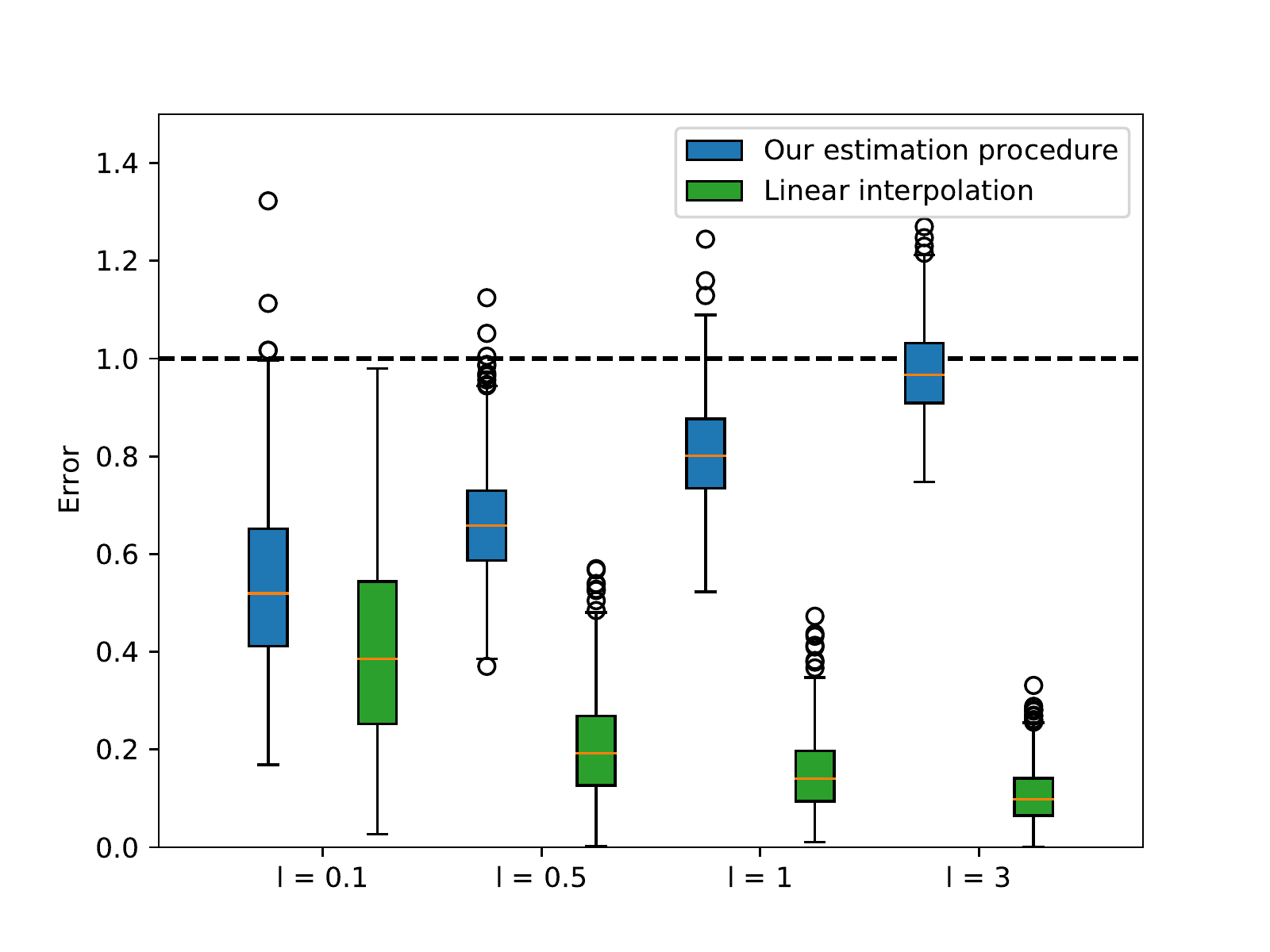}
		\caption{Results for run and tumble model}
		\label{fig:runandtumble}
	\end{subfigure}
	\caption{In these figures it can be seen how our estimation procedure performed compared to the straight-line estimator. In each figure it can be seen on the x-axis what choice of parameters is made and there are two boxplots for each choice. One that describes the estimation from our estimation procedure and one that comes from linear interpolation. For the internal state model there are $19$ outliers not visible for the first boxplot, the greatest outlier is $5.89$.}
	\label{fig:subfigures}
\end{figure}

In \Cref{fig:brownianmotion} we see that our estimator for the path length consistently performs a lot better than the straight line estimator. This is not unexpected, as in our model we assume the travel movement to be a realisation of a Brownian bridge. \Cref{fig:angularrandomwalk} shows that especially for large values of $\sigma$, our estimator performs better. Moreover we also see that our estimation procedure generally underestimates the path length. \Cref{fig:internalstate} looks a lot like \Cref{fig:brownianmotion}, but our estimator applied to random walks generated by the internal state model has more outliers, especially for low values of $S$. Lastly, \Cref{fig:runandtumble} shows a trend where our estimation procedure works better for higher values of $l$. Similarly to the \Cref{fig:angularrandomwalk} our estimation procedure seems to underestimate the path length.

We conclude that, as we expected, linear interpolation always underestimates the path length, and that this underestimation becomes greater for more random behaviour. 

\subsection{Radius of Gyration} \label{subsec:rog}
Similar comparisons have been performed regarding the RoG. First, we will quantify the error of the procedure as
\begin{equation*}
\text{error} = \frac{\text{RoG}(P_{\text{after}})}{\text{RoG}(P_{\text{before}})},\end{equation*} 
where $P_{\text{after}}$ is the path after filling the gap and $P_{\text{before}}$ is the simulated path before the gap. In Figure \ref{fig:rogprocedure}, the results are shown. Every error value has been rounded off to 1 decimal place. Figure \ref{fig:rogprocedure} shows histograms of different data generating models, where for each rounded error the number of simulations (out of 1000) with that certain error can be found. An error of $1$ implies that the RoG has been perfectly estimated. If the error is less than $1$, then the RoG is being underestimated. If the error is greater than $1$, then the RoG is being overestimated. 
\begin{figure}[h!]
     \centering
     \begin{subfigure}[b]{0.30\textwidth}
          \centering
          \includegraphics[width=\textwidth]{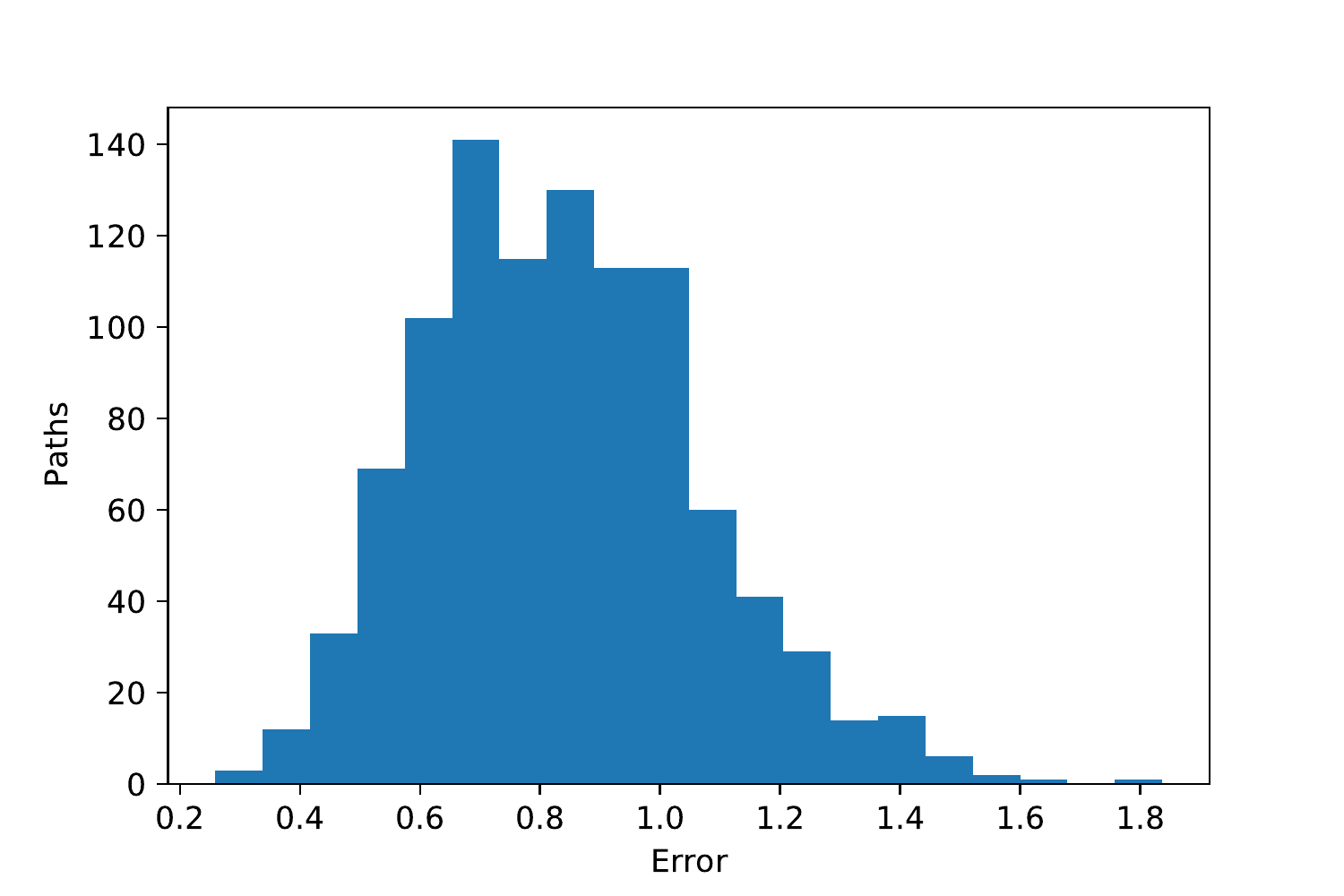}
          \caption{Fixed velocity random walk.}
          \label{fig:roghistofvrw}
     \end{subfigure}
     \hfill
     \begin{subfigure}[b]{0.30\textwidth}
         \centering
         \includegraphics[width=\textwidth]{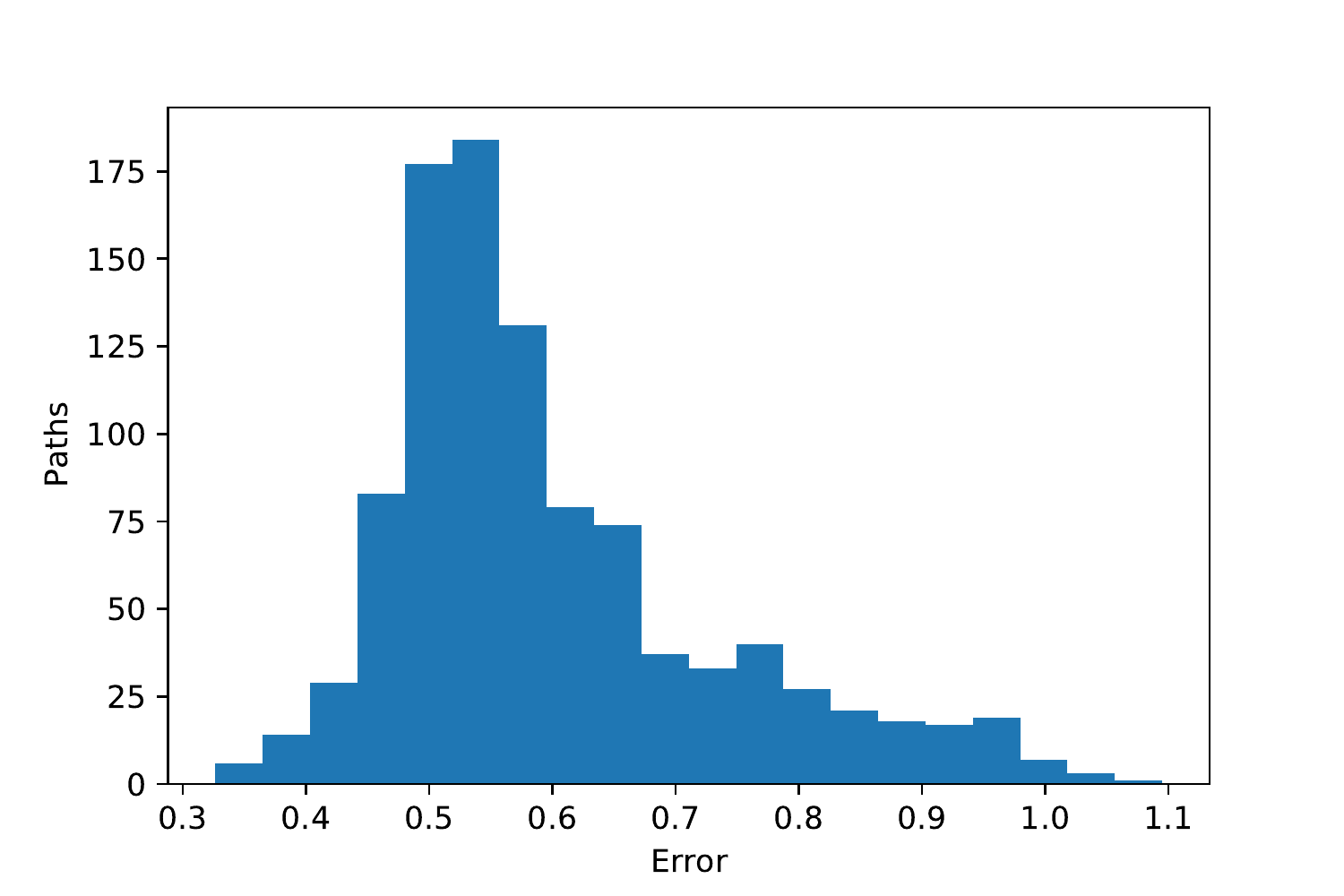}
         \caption{Fixed angular random walk.}
         \label{fig:roghistofarw}
     \end{subfigure}
     \hfill
     \begin{subfigure}[b]{0.30\textwidth}
         \centering
         \includegraphics[width=\textwidth]{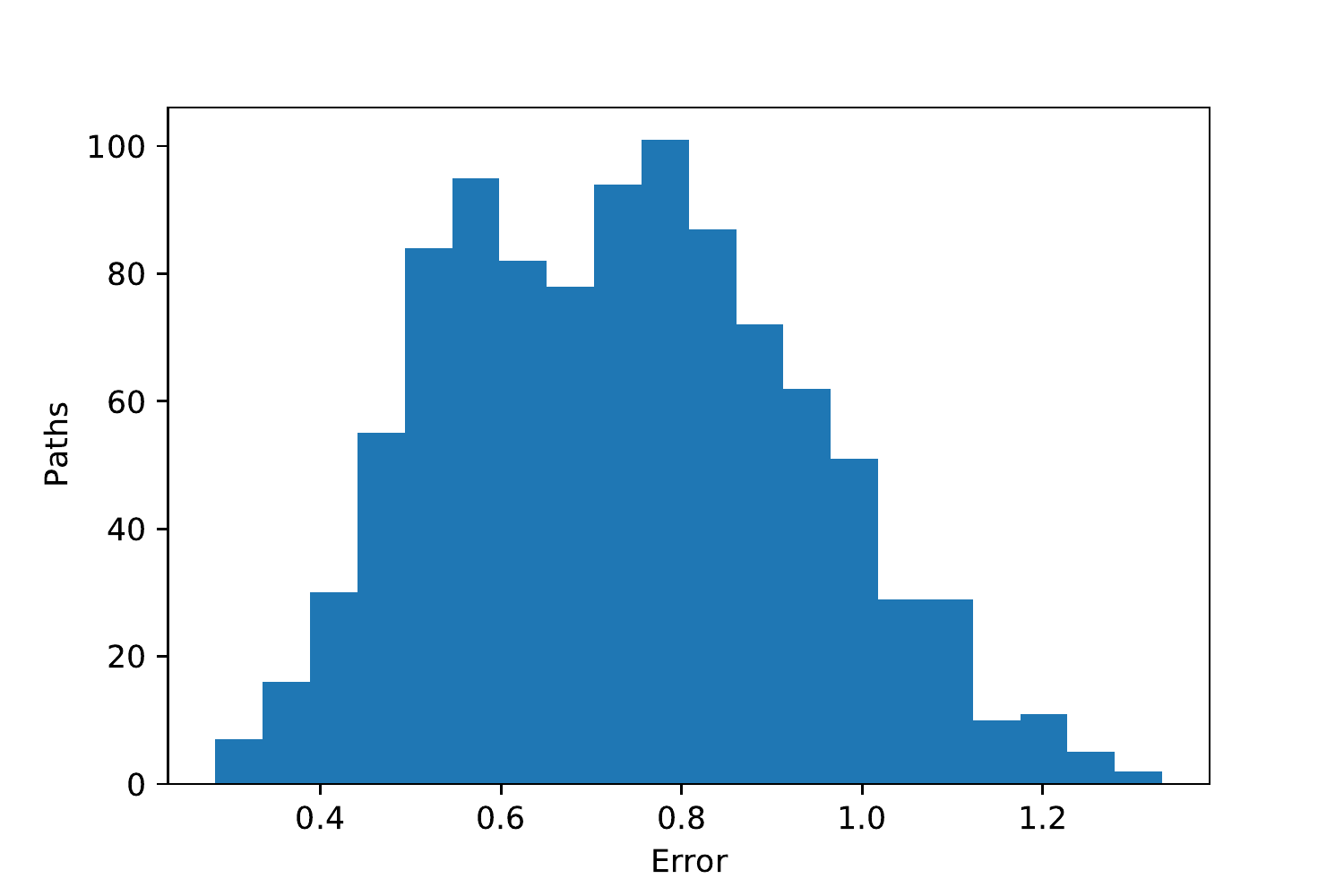}
         \caption{Run-and-tumble.}
         \label{fig:roghistorun}
     \end{subfigure}
    \caption{For every path generating model, a thousand paths have been simulated. For each path, we have omitted the first half of the location data and filled the gap by a Brownian bridge that follows from our estimation procedure. On the $x$-axis one can find the error and on the $y-$axis the number of paths with this error.}
    \label{fig:rogprocedure}
\end{figure}

For each model, the majority of the errors are less than $1$. This implies that we tend to underestimate the RoG. In general, interpolating by the Brownian bridge, gives a path that is more close to the centre point. There are certain outliers, where we have a strong underestimation or overestimation, and these are the cases where the missing part of the path behaves very differently in comparison with the part of the path that is available. Note that the half of the path is deleted and that indeed our data generating models can produce this random behaviour, where in the first half of the path we have a movement that is similar to a straight line and in the second half it could be that the path deviates a lot from the straight line. 

In \Cref{tab:rogresults}, we summarise these results. Again, one can see that on average, the RoG is being underestimated. We approach the paths that are obtained from the fixed velocity random walk the best in terms of RoG. However, there is a large deviation from the mean. Again, this shows how difficult it is to capture a trend when half of the path is missing. 
\begin{table}[h!]
    \centering
    \begin{tabular}{|c|c|c|c|c|}
    \hline
    & Mean $\text{RoG}(P_{\text{before}})$ & Mean $\text{RoG}(P_{\text{after}})$ & Mean error & Std. dev.\\
    \hline
    Fixed velocity & 12.300 & 9.807 & 0.841c& 0.231 \\
    \hline
    Fixed angular & 194.250 & 112.265 & 0.600 & 0.135 \\
    \hline
    Run-and-tumble & 25.693 & 18.130 & 0.737 & 0.198 \\
    \hline
    \end{tabular}
    \caption{The mean $\text{RoG}(P_{\text{before}})$, $\text{RoG}(P_{\text{after}})$, error and standard deviation of the errors over the $1000$ paths per model. Recall that the error is not a difference but a fraction that measures the $\text{RoG}(P_{\text{after}})$ relative to the actual RoG, i.e. $\text{RoG}(P_{\text{before}})$.}
    \label{tab:rogresults}
\end{table}

In the same way, we have obtained results by using linear interpolation to fill the gap, instead of our estimation procedure. In \Cref{fig:rogstraightline} these results are shown in similar fashion as above. 

\begin{figure}[H]
     \centering
     \begin{subfigure}[b]{0.30\textwidth}
          \centering
          \includegraphics[width=\textwidth]{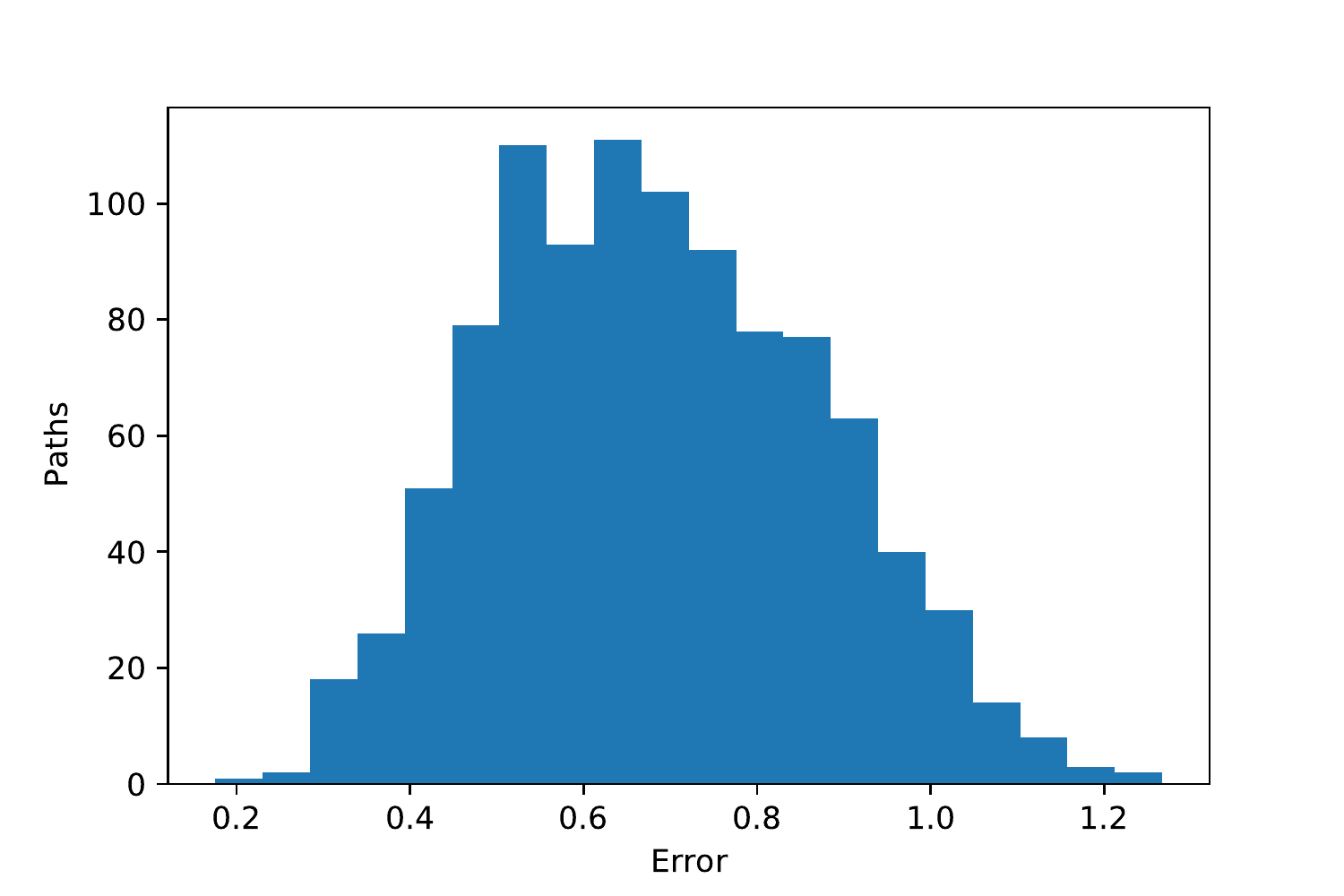}
          \caption{Fixed velocity random walk.}
          \label{fig:roghistofvrwstraight}
     \end{subfigure}
     \hfill
     \begin{subfigure}[b]{0.30\textwidth}
         \centering
         \includegraphics[width=\textwidth]{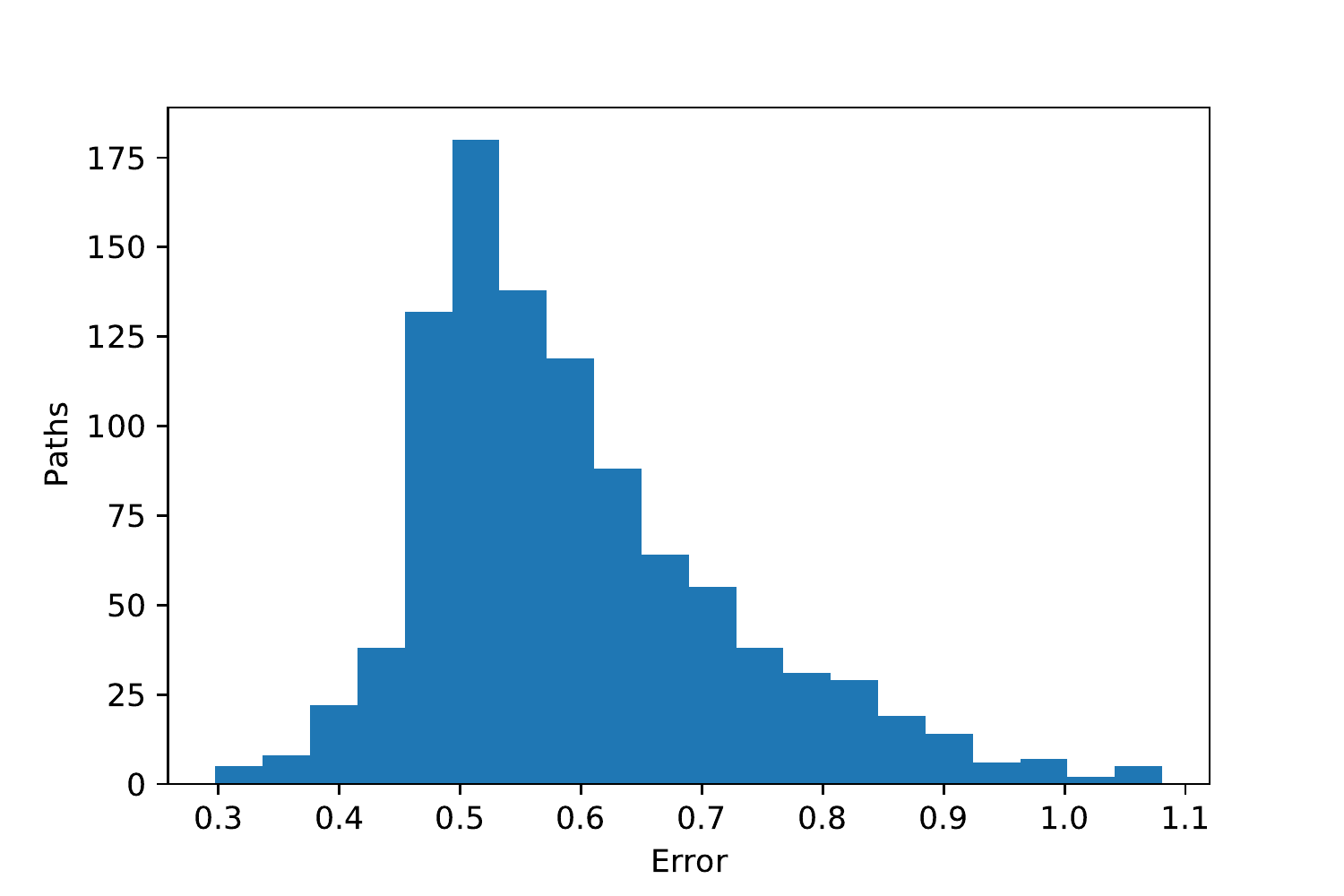}
         \caption{Fixed angular random walk.}
         \label{fig:roghistofarwstraight}
     \end{subfigure}
     \hfill
     \begin{subfigure}[b]{0.30\textwidth}
         \centering
         \includegraphics[width=\textwidth]{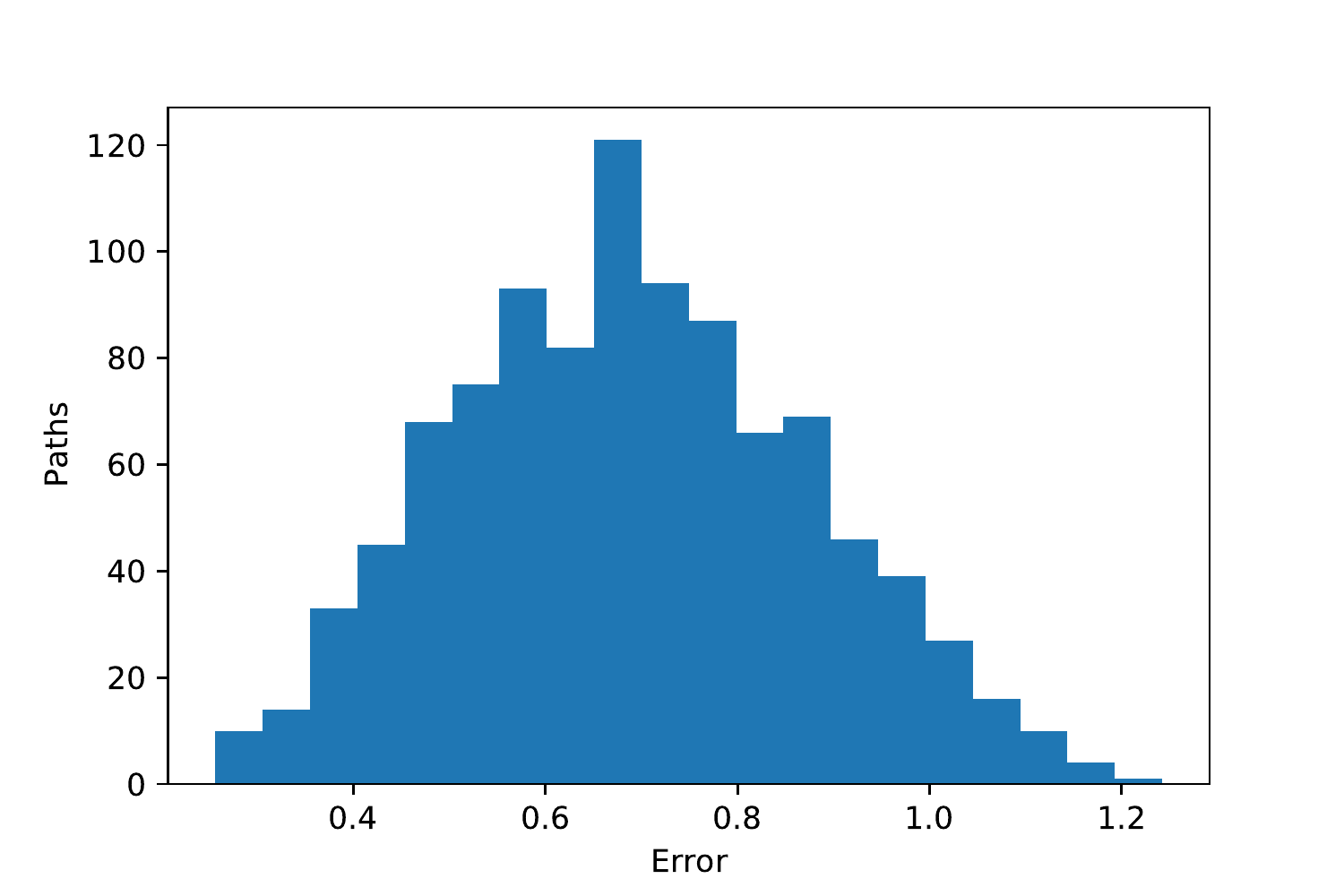}
         \caption{Run-and-tumble.}
         \label{fig:roghistorunstraight}
     \end{subfigure}
    \caption{For every path generating model, a thousand paths have been simulated. For each path, we have omitted the first half of the location data and filled the gap by a straight line from the origin to the end of the gap. On the $x$-axis one can find the error and on the $y-$axis the number of paths with this error. }
    \label{fig:rogstraightline}
\end{figure}

For each model, linear interpolation also tends to underestimate the RoG. However, we can also see that there is less overestimation. Besides this, the overestimation is limited to an error of $1.2$. 

In \Cref{tab:rogresultsstraight}, we summarise these results. Again, one can see that on average, the RoG is being underestimated. We approach the paths that are obtained from the run-and-tumble model the best in terms of RoG. However, for each model, the mean RoG after linear interpolation is less than the mean RoG after interpolating with a Brownian bridge. Also note that the deviation of the errors is less, so there is more stability in this case.  

\begin{table}[h!]
    \centering
    \begin{tabular}{|c|c|c|c|c|}
    \hline
    & Mean $\text{RoG}(P_{\text{before}})$ & Mean $\text{RoG}(P_{\text{after}})$ & Mean error & Std. dev.\\
    \hline
    Fixed velocity & 12.319 & 8.115 & 0.684 & 0.188 \\
    \hline
    Fixed angular & 196.910 & 113.531 & 0.594 & 0.130 \\
    \hline
    Run-and-tumble & 25.437 & 16.933 & 0.690 & 0.184 \\
    \hline
    \end{tabular}
    \caption{The mean $\text{RoG}(P_{\text{before}})$, $\text{RoG}(P_{\text{after}})$, error and standard deviation of the errors over the $1000$ paths per model. Note that $P_{after}$ in this context means: path after filling the gap by linear interpolation.}
    \label{tab:rogresultsstraight}
\end{table}

\clearpage
\section{Discussion}\label{sec:discussion}

\subsection{Interpretation of the Theoretical Results}

Suppose we have a dataset 
\begin{equation*}
    (z_0, t_0), \dots, (z_i, t_i), (z_{i+j+1}, t_{i+j+1}), \dots, (z_k, t_k),
\end{equation*}
so the data points $(z_{i+1}, t_{i+1}), \dots, (z_{i+j}, t_{i+j})$ are missing. Our first step would be to calculate an estimator for $\sigma_m$ based on the data that we do have, as explained in Section \ref{subsec:implementation}. Now, we have $T = t_{i+j+1} - t_i$ and $d = z_{i+j+1} - z_i$. For now, let us choose our discretisation such that we simulate exactly one point for each missing data point, so $n =T$. In this case, an estimator for the expected path length of the gap is given by
\begin{equation*}
    \hat{L} = \hat\sigma_m\sqrt{T(T-1)}\sqrt{\frac{\pi}{2}}L_{\frac{1}{2}}\left(-\frac{\norm{d}^{2}}{2\hat\sigma_{m}^{2}T(T-1)}\right), 
\end{equation*}
where $L_{\frac12}$ is a Laguerre function \cite{laguerre}. As we have seen in Section \ref{subsec:pathlength}, this estimator is bounded from below by $\norm{d}$. We think that $\hat{L}$ can be a more accurate estimator for the length of the path taken in the gap than $\norm{d}$. 

Note that it can be quite hard to estimate $\sigma_m$. First, one needs to make sure that the properties of the travel movement stay the same for all the points you use to estimate $\sigma_m$, and the points that are missing. This also means that in principle it does not matter how large the gap is, as long as the properties of the travel movement stay the same.

It can also be quite computationally expensive to calculate $\sigma_m$ for every single gap one encounters. Therefore, one could do empirical research into what value of $\sigma_m$ belongs to which travel movement. A couple of parameters $\sigma_m$ could depend on are the following:
\begin{itemize}
    \item Mode of transportation (walking, cycling, riding a car, taking the train, et cetera)
    \item Demographic properties of the respondent: age, medical condition, marital status, owning a car, owning a dog, address, job (and therefore educational background)
    \item Area of the travel movement (rural area, city centre, suburban, etc.)
    \item Temporal factors (time of day and day of the week)
    \item The weather
    \item Motivation and mood of traveller.
\end{itemize}

\subsection{Interpretation of the Experimental Results}

As expected the estimating procedure gives a great accuracy for estimation of data that was realised from the discrete Brownian motion model. This should not be a surprise and mainly shows that our estimating procedure is correctly implemented. Nonetheless, this is a natural form of motion that can certainly appear, so this is already a good sign for our procedure. The decrease of accuracy that the straight line estimation has for more chaotic behaviour does not appear for the accuracy of our procedure.

For the angular random walk model we see that there is a local minimum for the accuracy of our procedure. This mainly seems to be because our procedure increases in precision when more chaotic behaviour appears. Moreover almost straight line movement is also relatively easy to detect and estimate. There seem to be quite a few outliers for both our procedure and the straight line estimate. An explanation for these outliers was not found, but as there were $1000$ simulations some outliers can be expected.

For the random walk with internal state model we again see that our procedure seems to work a lot better than the straight line estimate. It is interesting however that this seems to be the only model where our procedure makes big overestimations. In the most extreme of the $1000$ simulations it overestimated the distance travelled by a factor $5.89$. It also seems to be that the straight line estimation does well on these outlier cases. A complete explanation for these outliers cannot be given at the moment, but this does correspond with the idea that this model has different possible states. If the path was moving around before and after the gap but had a long pause in the middle our estimation procedure does not factor in this idea of another state. 

For the run-and-tumble model it looks like our procedure does not work that well with this structure of behaviour. It did quite poorly for small variances and little chaos. As less of the run and tumble structure is visible because of higher variances and more random behaviour we see that our procedure starts to perform better. In all cases it did still outperform the straight line approach on average. 

Besides the path length, we have also obtained results regarding the RoG. In these results, we have defined an error that quantifies how well the RoG is being estimated. Again, our estimation procedure was compared with linear interpolation. On average, the estimation procedure seems to underestimate the RoG as the error is on average less than $1$. For the fixed velocity random walk, we see that the mean error comes closest to $1$. Therefore, the RoG is estimated the best for the fixed velocity random walk. The reason for this could be that the movement is more similar to the Brownian bridge. Intuitively, it has the same deviation from its expected path. Another aspect is the size of the RoG. It could also be that a larger RoG is harder to estimate than a smaller RoG. In other words, a path that has more points far away from the centre point is more difficult to estimate. This is in line with our intuition, as paths with unexpected behaviour are in general more difficult to comprehend. 

In conclusion, our estimates seem to work fine on paths that are modelled by some stochastic process. We have not looked at empirical data. However, we believe that our procedure can also be useful for such data. The gaps can be filled by the Brownian bridge with an estimated diffusion coefficient, and we can see that the path length of this Brownian bridge can be realistic: it comes close to the underlying path of the gap. For the radius of gyration, there are more dependencies which should be examined. One could also look at the RoG of the incomplete path and use this as a parameter for estimating the diffusion coefficient. This way, it could also be possible to tackle incomplete paths with a larger RoG.

\subsection{Possible Flaws in the Experiment}

Quite an important difficulty in the test of our estimating procedure is that we were not able to use real location data of actual travellers. All of our experimentation has been performed on artificial data, which we could not compare to real life data, so we cannot say what kind of models depict real life best. For example, we did not model change in mode of transport, but we do expect this quite regularly in real-life travel movements. A person walking to their car and then driving gives quite a different estimation for their diffusion coefficient in time. In our procedure this is averaged over all available data points which gives a less complete picture. Moreover, there are a lot of reasonable assumptions you can make for real life movement, such as a velocity or turn-frequency dependent on the mode of travel. Our procedure does not take into account any of such assumptions.

Another more theoretical problem with our procedure is that Brownian motion is highly dependent on the discretisation of the model. If we would take a finer discretisation, the expectation value of the path length would increase. Our procedure is bases on a discretisation with one point every second. This is an arbitrary choice that does match up with actual measurements. One way to avoid this problem would be to use our estimation procedure to calculate the expected length for the intervals that we do have complete data of. Then we can use the error from these known intervals to tweak the expected value for the gap.

Our procedure works by estimating the variance as if the given movement was Brownian motion. The variance could have been estimated differently, for example dependent on the scale and discretisation.

\clearpage
\section{Conclusion and Outlook}\label{sec:conclusion}

In summary, we have used the Brownian Bridge Movement Model to model a person's movement and to fill up gaps between measurements. Through a maximum likelihood estimator the relevant parameter is can be estimated. By discretisation of a Brownian bridge, we have found a distribution for the path length. This distribution is given by Equation \ref{eq:discpathlength}. 

To test the Brownian bridge model, we have applied it to gaps in four numerically simulated processes meant to mimic travel data. Then, we estimated two properties of the simulated paths. The length of a part of the path and the radius of gyration. 
From these numerical results we can conclude we consistently estimate path length better in comparison to a straight line. This means that when we would use the Brownian bridge as a way to fill a gap, it would have a better estimate of the length of the sub-path (that has been deleted) than a straight line procedure. This is a meaningful result, as it could indicate that Brownian bridge interpolation is more realistic than linear interpolation. 

We have also measured the estimation procedure in terms of RoG. Again, we have compared the estimation procedure with linear interpolation and the results show that the estimation procedure has a better estimate for the RoG than that of a straight line. However, we still strongly underestimate the actual RoG. Thus, there seems to be a small improvement, but there is still room for more. 

\subsection{Outlook}

First of all, due to privacy regulations we could not access any data from actual people travelling, so it is natural to suggest for future work the CBS verifies our results with their data.

The CBS could also look into the sample rate the app uses.
Perhaps equally good results can be achieved by lowering the sample rate, or even making the rate dynamic.
For instance, a high sample rate could be used in cases where the diffusion coefficient is high, meaning the travel behaviour is less predictable, and a low rate in cases where the travel behaviour is very predictable and has a low diffusion coefficient.
A consequence of this could be that the app uses less battery power, making the amount of data per survey increase both because the internal software is less likely to shut the app to preserve battery power and because respondents are more likely to continue with the survey if it is less of a battery strain.

In particular, future work can be done on estimating the diffusion coefficient in various cases. For instance, empirical data can be used to decide on which variables the diffusion coefficient depends. Possible variables could be the mode of transportation, the demographic properties of the respondent, the area of the travel movement, temporal factors or the weather.

\clearpage
\bibliographystyle{plain}
\bibliography{bibfile}
\addcontentsline{toc}{section}{References}

\clearpage
\appendix
\section{Useful Mathematical Background}\label{app:background}

In this section, several definitions and theorems that were used in the report are stated. 

\begin{proposition}[Sum of two independent normal random variables]\cite[pg. 158]{IKSboek} \label{prop:sumnormals}
Let $X,Y$ be two independent random variables where $X \sim N(\mu_1,\sigma^{2}_1)$ and
$Y \sim N(\mu_2,\sigma^{2}_2)$. Then the sum $X+Y$ is again normally distributed: 
\begin{equation*}
    X + Y \sim N(\mu_1 + \mu_2, \sigma^{2}_1 + \sigma_{2}^{2}). 
\end{equation*}
\end{proposition}

\begin{definition}
[Probability density function of a bivariate normal distributed random variable]\label{def:pdf2} A random vector $Z$ with bivariate normal distribution $N_2(\mu, \sigma^2\I_2)$ has the probability density function
\begin{equation*}
    f: \R^2 \to \R, \;\;\; z \mapsto \frac{1}{2\pi\sigma^2} \exp \left( -\frac{\norm{z - \mu}^2}{2\sigma^2} \right).
\end{equation*}
\end{definition}

\begin{proposition}[Distribution of two independent marginal normal random variables]\label{prop:distrnormals}
Let $X,Y$ be two independent random variables where $X \sim N(\mu_1,\sigma^{2}_1)$ and
$Y \sim N(\mu_2,\sigma^{2}_2)$.
Then the vector $(X, Y)$ has a bivariate normal distribution $N_2(\mu, \sigma^2\I_2)$, where 
\begin{equation*}
    \mu = \begin{pmatrix}
        \mu_1\\
        \mu_2
    \end{pmatrix} \quad \text{ and } \quad \sigma^{2} = \begin{pmatrix}
        \sigma_1^{2}\\
        \sigma_2^{2}
    \end{pmatrix}.
\end{equation*}
\end{proposition}
\begin{proof}
As $X$ and $Y$ are independent, their joint probability density function can be written as
\begin{align*}
    f_{X, Y}(x, y) &= f_X(x) f_Y(y) = \frac1{2\pi \sigma^2} \exp \left(-\frac{(x-\mu_1)^2 + (y-\mu_2)^2}{2\sigma^2} \right) \\
    &= \frac{1}{2\pi\sigma^2} \exp \left( -\frac{\norm{(x, y) - (\mu_1, \mu_2)}^2}{2\sigma^2} \right).\qedhere
\end{align*}
\end{proof}

\begin{definition}[Brownian motion] \label{def:brownianmotion}\cite[pg. 94]{shrevenfinancialmath} 
A (standard) Brownian motion or Wiener process is a stochastic process $W = \{W_t: t \in \left[0, \infty\right)\}$ such that
\begin{enumerate}
    \item $\P(W_0 = 0) = 1$;
    \item if $0 \leqslant s < t$, then $W_t - W_s \sim W_{t-s}$; \footnote{Throughout this chapter $ X \sim Y$ if $X$ is distributed as $Y$. }
    \item if $0 \leqslant t_1 < t_2 < \ldots < t_n$, then the increments $W_{t_1}, W_{t_2} - W_{t_1}, \dots, W_{t_n} - W_{t_{n-1}}$ are independent;
    \item if $t > 0$, then $W_t \sim N(0, t)$;
    \item the mapping $t \mapsto W_t$ is almost surely continuous.
\end{enumerate}
\end{definition}

\begin{proposition}\label{prop:1dbridgedifference}
Let $X = \{X_t : t \geqslant 0\}$ be a one-dimensional standard Brownian bridge. Then $X_t - X_s \sim X_{t-s}$ for all $0 \leqslant s < t \leqslant T$.
\end{proposition}
\begin{proof}
Let us write
\begin{equation*}
X_t = W_t + \frac{t}{T} W_T,
\end{equation*}
where $W$ is a standard Brownian motion. Now we can calculate
\begin{align*}
    X_t - X_s
    &=  W_t - \frac t T W_T - W_s + \frac{s}{T} W_T \\
    &= (W_t - W_s) - \frac{t-s}{T} W_T.
\end{align*}
Since $W$ is a Brownian motion we have $W_t - W_s \sim W_{t-s}$, so
\begin{equation*}
    X_t - X_s \sim  W_{t-s} - \frac{t-s}{T} W_T = X_{t-s}. \qedhere
\end{equation*}
\end{proof}

\begin{definition}[Rayleigh distribution] \label{def:rayleighdistribution} \cite[pg. 276]{Riciandistribution}
Let $X, Y \sim N(0,\sigma^{2})$. Then $R = \sqrt{X^{2} + Y^{2}}.$ has the Rayleigh distribution, with the following mean: 
\begin{equation*}
    \mathbb{E}[R] = \sqrt{\frac{\pi}{2}} \sigma.
\end{equation*}

\end{definition}
\begin{definition}[Rice distribution] \label{def:riciandistribution} \cite[pg. 278]{Riciandistribution}
Let $Z \sim N_2(\mu, \sigma^2\I_2)$. Then $\norm{Z}$ has the Rician distribution with parameters $A = \norm{\mu}$ and $B = \sigma^2$. The expected value of $R$ is given by
\begin{equation*}
    \E[R] = \sqrt{B} \sqrt{\frac{\pi}{2}}L_{\frac{1}{2}}\left( - \frac{A^{2}}{2B} \right).
\end{equation*}
\end{definition}

\section{Simulated processes}\label{app:processes}

For our experiments, we have used data from five different simulated stochastic processes. We will now describe in more detail how these processes work.

We have used a range of simulated stochastic process to test the performance of the Brownian bridge method. The simplest model we considered, was a random walk with normally distributed increments. This process at time step $n$ is given by the sum $n$ of normally distributed random variables with a variance $\sigma$. We will refer to this model as \emph{discrete Brownian motion}. Quite similarly, we considered a process where we normalised the increments to all be the same distance. We refer to this model as the \emph{fixed velocity random walk}. 

The next model we considered, was a process that moves at some fixed speed $v$, but its movement direction changes. In this case, the angle that specifies the movement direction in $\mathbb{R}^2$ is a sum of independent normally distributed random variables with variance $\sigma$. We will refer to this model as an \emph{angular random walk}. For very high variances of these random variables the direction for some time step will almost be independent of the previous direction and therefore this behaviour will be similar to the fixed velocity random walk. 

Furthermore, we considered a random walk model which keeps track of an internal state which affects the chance to move in a certain direction. This model has a moving agent with a predetermined probability for the agent to either change direction on a grid or start/stop moving at each time step. The agent then either moves a set length in the determined direction or does not move. Ideally the chosen probabilities are gained from real movement data though we had to make educated guesses since we did not have access to this data. For our probabilities we used the assumptions that there is no inherent preference for turning left or right, that the chance of someone not changing state is large, and that the chance of suddenly making a 180 degree turn is small. We will refer to this model as the \emph{random walk with internal state}.

We also used this model to make a bridge between two points by generating random paths until one got to the correct endpoint. Even though this model had some nice properties the computational power needed to find a bridge this way quickly becomes huge if the bridge is long. This is why we stopped using this method, the code for this model can be found in our \href{https://github.com/kerimdelic/MathForIndustry/blob/main/CBS.ipynb}{notebook}. 

Finally, we also considered a \emph{run-and-tumble model}. Most of the time the particle moves in a straight line. However, at each time step, there is a probability of $1 - e^{-l}$ with $l > 0$ that the movement direction is changed. In case the movement direction is changed, a new direction is chosen with uniform probability. The run part refers to the straight line movement, whilst the tumble refers the random change of direction. We will refer to this model as a run-and-tumble.

\end{document}